\documentclass[10pt,a4paper]{article}

\usepackage{a4wide}
\usepackage[T1]{fontenc} % codifica dei font in uscita
\usepackage{lmodern}
\usepackage[latin1]{inputenc} % lettere accentate da tastiera
\usepackage[english]{babel} % lingua principale del documento
\usepackage{graphicx} %per inserire le immagini usare il comando \includegraphics[hopzionei=hvalorei,h. . .i]{himmaginei}
\graphicspath{{figures/}}
\DeclareGraphicsExtensions{.pdf}
\usepackage{amsmath,amssymb,amsthm,amsfonts,mathrsfs}
\usepackage{authblk}  %per gli autori
\usepackage{color}
\usepackage{subfig}
\usepackage{caption}
\usepackage{algorithmic}
\usepackage{algorithm}

\usepackage[subnum]{cases}
\usepackage{todonotes}

\newcommand{\vm}{V_{\max}}
\newcommand{\rhom}{\rho_{\max}}
\newcommand{\vb}{v_*}
\newcommand{\va}{v^*}

\newcommand{\V}{\mathcal{V}}
\newcommand{\Prob}[1]{\mathcal{P}(\vb\to v|#1;\rho)}
\newcommand{\Dv}{\Delta v}

\newcommand{\abs}[1]{\left\vert#1\right\vert}

\newcommand{\norm}[1]{\left\lVert#1\right\rVert}

\DeclareMathOperator{\Acc}{\mathsf{Acc}}
\DeclareMathOperator{\Brak}{\mathsf{Brak}}

\newcounter{the}
\newtheorem{theorem}[the]{Theorem}
\newtheorem{remark}[the]{Remark}

\author{Giuseppe Visconti \\ \vspace{-3mm}
	{\small\it Department of Science and High Technology} \\ \vspace{-1mm}
	{\small\it University of Insubria} \\ \vspace{-1mm}
	{\small\it Via Valleggio 11, 22100 Como, Italy}\\
	\vspace{5mm}
	Michael Herty \\
	{\small\it Department of Mathematics} \\ \vspace{-1mm}
	{\small\it RWTH Aachen University} \\ \vspace{-1mm}
	{\small\it Templergraben 55, 52062 Aachen, Germany} \\ \vspace{5mm}
	Gabriella Puppo \\
	{\small\it Department of Science and High Technology} \\ \vspace{-1mm}
	{\small\it University of Insubria} \\ \vspace{-1mm}
	{\small\it Via Valleggio 11, 22100 Como, Italy}\\ \vspace{5mm}
	Andrea Tosin \\
	{\small\it Department of Mathematical Sciences ``G. L. Lagrange''} \\ \vspace{-1mm}
	{\small\it Polytechnic University of Turin} \\ \vspace{-1mm}
	{\small\it Corso Duca degli Abruzzi 24, 10129 Turin, Italy}
}

\date{}

\begin{document}

\title{Multivalued fundamental diagrams of traffic flow in the kinetic Fokker-Planck limit}
\maketitle

\begin{abstract}
	Starting from interaction rules based on two levels of stochasticity we study the influence of the microscopic dynamics on the macroscopic properties of vehicular flow. In particular, we study the qualitative structure of the resulting flux-density and speed-density diagrams for different choices of the desired speeds. We are able to recover multivalued diagrams as a result of the existence of a one-parameter family of stationary distributions, whose expression is analytically found by means of a Fokker-Planck approximation of the initial Boltzmann-type model. 	
\end{abstract}

\paragraph{Keywords} Traffic modeling, kinetic models, Boltzmann, Fokker-Planck limit, multivalued diagrams 

\paragraph{MSC} 90B20, 65Z05, 35Q20, 35Q70, 35Q84

\section{Introduction} \label{sec:Introduction}

In traffic flow theory a basic tool to analyze vehicular traffic problems is given by the so-called fundamental relations among the macroscopic quantities of the flow. Mainly two of them are used: the flux-density diagram (or fundamental diagram) and the speed-density diagram. They allow one to study the macroscopic trends of vehicle dynamics as a result of the interactions taking place at the microscopic scale. Therefore, macroscopic relations can be used for instance to predict the maximum capacity of a road when certain traffic rules are imposed, such as e.g. inflow block of heavy vehicles or speed limits. Since the macroscopic dynamics depend on different properties of the flow, such as the types of vehicles traveling on the road or the external conditions (e.g. weather or environmental conditions), experimental diagrams show a multivalued structure: several values of the flux of vehicles or of the mean speed may correspond to a given value of the vehicle density, see for instance Figure \ref{fig:Comparisons} in Section \ref{sec:FundamentalDiagrams} with experimental data published in \cite{seibold2013NHM}. The qualitative structure of such diagrams is defined by the properties of different regimes, or phases, of traffic as
widely studied, e.g. in \cite{GreenshieldsSymposium,kerner2004BOOK}.

Mathematical models for traffic flow capable of providing fundamental diagrams which reproduce the structure of experimental data play thus an important role. In the current mathematical literature, there are three different approaches to model traffic flow phenomena, which differ from each other in the scale of representation. {\em Microscopic} models look at the motion of each single vehicle, predicting the evolution of its position and speed by means of a system of ordinary second-order differential equations, \cite{aw2002SIAP,KernerKlenov,ZhangMultiphase}. Clearly, the larger is the number of vehicles, the more computationally expensive is the model. On the opposite end, {\em macroscopic} models provide a large-scale aggregate point of view, neglecting the microscopic dynamics. Therefore, they study only the evolution of the macroscopic quantities related to the traffic flow by means of partial differential equations inspired by fluid dynamics laws, \cite{aw2000SIAP,benzoni2003EJAM,lighthill1955PRSL,richards1956OR}. Although, in contrast to microscopic models, macroscopic equations can be easily treated, it is necessary to complete the equations with a closure law. This means that the fundamental diagram is given as an a priori relation derived from heuristic or physical arguments, which therefore usually do not result from microscopic dynamics, see however \cite{klar1997Enskog}.

In this paper we focus instead on kinetic models ({\em mesoscopic} scale), which are based on a statistical mechanics approach. They provide an aggregate description of traffic flow by linking collective dynamics to microscopic interactions among the vehicles. These models are thus characterized by a statistical function giving the distribution of the velocity of the vehicles on the road. Kinetic models provide quite naturally a closure law based on the modeling of microscopic dynamics, thus without assuming previous knowledge on the dependence of the mean velocity on the local density of
traffic. In fact, macroscopic quantities are computed as moments, with respect to the microscopic speed, of the distribution function. Several kinetic approaches have been proposed, starting from the pioneering work \cite{Prigogine61,PrigogineHerman} and later \cite{paveri1975TR}. Recently, kinetic models have been widely used, e.g. to model multilane traffic flow \cite{klar1999SIAP-1,klar1999SIAP-2}, flows on networks \cite{fermotosin2015}, homogeneous space problems \cite{HertyKlarPareschi}, inhomogeneous space phenomena with non-local interactions \cite{klar1997Enskog}, heterogeneous compositions of the flow \cite{PgSmTaVg}. These models are based on a Boltzmann-type collision term, which models speed changes through a transition probability depending on the local traffic conditions and describing the relaxation in time of the kinetic distribution due to the interactions among vehicles. Although in \cite{PgSmTaVg2} the analytical expression of the stationary solution is found using a suitable choice of the microscopic interactions, in general the interaction integrals appearing in kinetic Boltzmann-type models for traffic flow typically do not provide explicitly the equilibrium distribution and they are very demanding from a computational point of view, see \cite{KlarWegener96}. For this reason, two main approaches have been taken into account
in order to compute the time-asymptotic distribution or to reduce the computational cost: on the one hand, one may consider Vlasov-Fokker-Planck-type models in which the interaction integrals
are replaced by differential operators, see \cite{HertyIllner08,HertyPareschi10,IllnerKlarMaterne}; on the other hand, one may consider simplified kinetic models with a small number of velocities, namely the discrete-velocity models, see e.g. \cite{FermoTosin14,Herty2007481}.

In this paper, starting from a Boltzmann-type model characterized by microscopic interaction rules accounting for two levels of stochasticity in drivers' subjective decisions, we recover a Fokker-Planck-type kinetic model as a {\em grazing collision limit}, see e.g. \cite{Desvillettes,DiPernaLions,PareschiToscani2006,PareschiToscaniVillani,Villani1999}. In particular we follow the technique introduced in \cite{HertyPareschi10} for traffic flow, or in \cite{Toscani2006} in the context of opinion formation models and in \cite{PareschiToscani2006} for gas-dynamics. On the one hand, the goal is to find the analytical expression of the steady state distributions starting from the microscopic dynamics. On the other hand, we are interested in reproducing multivalued diagrams without using the multi-population framework \cite{PgSmTaVg}, in which the scattering behavior of experimental data is explained as a result of the heterogeneous composition of the flow of vehicles on the road. This approach requires the solution of a system of kinetic equations describing the evolution of each kinetic distribution associated to each class of vehicles. Here, instead, assuming that drivers react only to average quantities, multivalued diagrams are obtained as a result of the existence of a one-parameter family of stationary solutions. We observe by numerical evidence that this parameter is strictly linked to the macroscopic characteristics of the flow, as the critical density, the maximum flux, the maximum speed in free-flow regime and the maximum density. This means that the parameter synthesizes the different properties of the flow which can induce different macroscopic dynamics. We also improve the multivalued diagrams obtained in \cite{IllnerKlarMaterne}, in which the authors introduce a Fokker-Planck-type multilane model and recover scattered data as a consequence of the introduction of a probability of lane changing. As a matter of fact, the diagrams show only a small range of density in which the flux of vehicles is multivalued. In addition to that, they reproduce a capacity drop which is too sharp.
%Moreover, in contrast to \cite{PgSmTaVg2} in which the formulation of the transition probability produces quantized equilibrium solutions on a reduced space of velocity, defining thus a link between continuous-velocity and discrete-velocity models, the steady state distributions of the present model, instead, provides a richer and more realistic information on the distribution of the vehicles' velocity on the road.\todo{Gabriella trova masochistica questa frase perch\'e citiamo il nostro lavoro in termini critici.}

Several Fokker-Planck-type models for traffic flow have been proposed in the literature. For instance in \cite{HertyPareschi10} the authors derive a Fokker-Planck-type model as an asymptotic limit of a Boltzmann-type model featuring microscopic interactions in which the drivers react to the mean speed. However, in \cite{HertyPareschi10} the authors are mainly interested in establishing links between Fokker-Planck-type models based on a heuristic modeling of braking/acceleration and diffusive operators, see \cite{IllnerKlarMaterne}, and Fokker-Planck type models obtained as limit of Boltzmann-type equations, thus based on microscopic motivations for the given form of the operators. In this sense, the present work can be considered as a natural sequel to \cite{HertyPareschi10}. In fact, on the one hand, here we rewrite the microscopic rules of \cite{HertyPareschi10} in order to take into account the influence of the road congestion on the drivers' decisions. Furthermore, we simplify the grazing collision limit, which now does not depend on free parameters. On the other hand, we focus mainly on the study of macroscopic trends of traffic by means of the fundamental diagrams.

The paper is organized as follows. In Section \ref{sec:Boltzmann} we introduce the microscopic interaction rules and we discuss the similarities with other models. In particular, we show that the rules can be obtained by averaging the post-interaction speeds prescribed in \cite{PgSmTaVg2} with the addition of a stochastic perturbation due to the inability of the drivers to adjust precisely their speed. In contrast to \cite{HertyPareschi10}, another stochastic term is introduced, which models the influence of the road congestion on the drivers' behavior. In Section \ref{sec:FokkerPlanck} we derive the Fokker-Planck model as the grazing collision limit of the Boltzmann-type equation. In Section \ref{sec:FundamentalDiagrams} we recover the time-asymptotic kinetic distributions in general form. Then, we specify the interaction rules by choosing the desired speeds of the drivers in braking and acceleration scenarios and we study the fundamental diagrams resulting from two cases. We show that a particular choice of the desired speeds allows one to reproduce the closures law for macroscopic models, as for instance the Greenshields' closure \cite{greenshields}. Conversely the diagrams obtained from the second set of desired speeds reproduce very well the qualitative structure of experimental data, in particular their scattering in the congested flow regime. Finally, we end in Section \ref{sec:Conclusions} summarizing the main results of this work, and proposing possible applications and further developments.

\section{Boltzmann model} \label{sec:Boltzmann}

Recalling that the core of a kinetic model for traffic flow is the modeling of the microscopic interactions among the vehicles, we first discuss the choice of a set of interaction rules and then we recover a Boltzmann-type equation. As usual, the model is characterized by a collision term which describes the relaxation of the kinetic distribution towards equilibrium.

Here, we assume that drivers react to the mean speed of the surrounding vehicles. Thus, they adjust their velocity by comparing their actual speed with the macroscopic velocity $u$. We recall that, in kinetic theory, the macroscopic variables are recovered as moments of the kinetic distribution function $f=f(t,v):  \mathbb{R}^+ \times \V \to \mathbb{R}^+$. Here we assume that $\V=[0,\vm]$ is the space of the microscopic velocities, $\vm>0$ being the maximum speed, which usually depends on the mechanical characteristics of the vehicles, on imposed speed limits or on the type of drivers. %, according to the maximum velocity they intend to keep in free road conditions.
Then
\[
	\rho(t)=\int_{0}^{\vm} f(t,v) \mathrm{d}v, \quad (\rho u)(t)=\int_{0}^{\vm} v f(t,v) \mathrm{d}v.
\]

\begin{remark}
	Throughout the paper we consider the \textquotedblleft spatially homogeneous case\textquotedblright, i.e. we assume that vehicles are uniformly distributed along the road so that the kinetic distribution function $f$ depends only on their microscopic speed $v\in\V$ but not on their position $x$. This assumption allows for a direct focus on the interaction dynamics among the vehicles, neglecting possible space inhomogeneities.
\end{remark}

\subsection{Microscopic model} \label{sec:MicroscopicModel} Let $\vb$ be the pre-interaction velocity of a vehicle and let $u$ be the mean speed of the flow. We distinguish two cases:

\begin{description}
	\item[Acceleration] Vehicles tend to increase their speed when they are slower than the speed of the flow, namely if $\vb < u$;
	\item[Braking] Vehicles tend to decrease their speed when they are faster than the speed of the flow, namely if $\vb > u$.
\end{description}

\begin{remark}
	Notice that we do not model the scenario in which $\vb = u$, when it is realistic to assume that vehicles do not change their speed. We will see that this case does not give contribution in the Fokker-Planck approximation, see Section \ref{sec:FokkerPlanck}.
\end{remark}

We point out that the interaction rules in traffic flow are different from the gas dynamics case. In particular, the post-interaction speeds resulting from the two scenarios described above are not necessarily symmetric. Therefore, from a microscopic point of view, we can prescribe the following form of the output velocity $v$ resulting from acceleration and braking:

\begin{subequations}
	\label{eq:GenericMicroRules}
	\begin{align}
		v = \vb + \epsilon \, P(\rho) \, \Dv_A(\vb,\rho) + \sqrt{\epsilon \, Q_A(\rho)} \, \Dv_A(\vb,\rho) \, \xi, \quad &\text{if \; $\vb < u$} \label{eq:GenericMicroRulesAcc}\\
		v = \vb - \epsilon \, (1-P(\rho)) \, \Dv_B(\vb,u,\rho) + \sqrt{\epsilon \, Q_B(\rho)} \, \Dv_B(\vb,u,\rho) \, \xi, \quad &\text{if \; $\vb > u$} \label{eq:GenericMicroRulesBrak}
	\end{align}
\end{subequations}
where
\begin{itemize}
	\item $\Dv_A(\vb,\rho)$ is the jump of velocity in the acceleration interaction. We assume that
	\[
		\Dv_A(\vb,\rho):=V_A(\vb,\rho)-\vb.
	\]
	The quantity $V_A(\vb,\rho)\in(\vb,\vm]$ is the desired speed of a vehicle that increases its velocity $\vb$. In a realistic framework, $V_A(\vb,\rho)$ cannot be taken as a fixed parameter. Therefore, in contrast to \cite{HertyPareschi10}, we take $V_A(\vb,\rho)$ dependent on the actual speed $\vb$, in order to preserve the bounds of the acceleration, see \cite{Lebacque03}. Potentially, one can consider the case in which the desired speed $V_A(\vb,\rho)$ depends also on the density of vehicles $\rho$, see Section \ref{sec:Case1}.
	%	The fact that $\Dv_A$ does not depend on the mean speed $u$ means that $\Dv_A$ is constant with respect to the distance between the actual speed $\vb$ and the mean speed.
	%
	% Is this sufficient to ensure that the acceleration of the model is bounded? Remember that in \cite{HertyPareschi10} $\Dv_A=u-\vb$: may this produce unbounded accelerations?
	%
	\item $\Dv_B$, instead, is the jump of velocity in the braking interaction. As discussed for $\Dv_A(\vb,\rho)$, we consider
	\[
		\Dv_B(\vb,u,\rho):=\vb-V_B(u,\rho).
	\]
	The quantity $V_B(u,\rho)\in[0,\vb)$ is the desired speed of a vehicle which is decreasing its velocity $\vb$. Notice that, in contrast to the desired speed in acceleration, $V_B(u,\rho)$ is not a function of $\vb$ but it depends on the macroscopic flux, in particular on the local density $\rho$ and on the mean speed $u$. In fact, in the braking scenario, the desired speed is mainly influenced by the conditions of the traffic flow. For instance, let us consider the case in which a driver bumps into a jam. Initially the driver travels with a velocity larger than the the speed of the flow in front of him, i.e. $\vb \gg u$. Then, the driver must brake and adjust his speed to the traffic condition ahead independently of the value of his actual velocity $\vb$. Therefore we expect that the jump of velocity $\Dv_B(\vb,u,\rho)$ is greater when the difference between the actual speed and the mean speed is larger, as for instance in \cite{HertyPareschi10};
	%Remember that in \cite{HertyPareschi10} $\Dv_B=\vb-u$.}
	%
	\item $\xi$ is a random variable with given distribution $\eta(\xi)$, having zero mean and variance $\sigma^2$. The presence of $\xi$ allows one to include a noise term proportional to the velocity jump. From the modeling point of view, the noise term is introduced in order to consider the fact that drivers are not able to estimate and to reach precisely the desired speed after acceleration and braking, as prescribed in \eqref{eq:GenericMicroRules};
	% Instead, from a mathematical point of view, using the time scaling considered in Section \ref{sec:FokkerPlanck} for the grazing collision limit, the diffusive part of the Fokker-Planck equation will result only if the noise term is modeled;
	%
	\item $\epsilon \in (0,1)$ is a deterministic and dimensionless parameter which models the time scaling. We notice that for $\epsilon \to 0$ we have $v \to \vb$ for both $\vb < u$ and $\vb > u$. This means that $\epsilon$ controls also the strength of the microscopic interactions;
	\item $P(\rho)$ is the {\em probability of accelerating}, see e.g. \cite{KlarWegener96,Prigogine61}, taken as a function of the local density $\rho$. It allows one to define the post-interaction speeds as functions of the density. Thus, the post-interaction speed $v$ defined by the rules \eqref{eq:GenericMicroRules} is not deterministically fixed. Although several laws of the probability function $P(\rho)$ are taken into account in order to better fix experimental data, see e.g. \cite{Prigogine61,PrigogineHerman,PgSmTaVg3}, we will consider the simplest choice
	\begin{equation}\label{eq:Plaw}
		P(\rho)=1-\frac{\rho}{\rhom},
	\end{equation}
	where $\rhom$ is the maximum density taken as the maximum number of vehicles per unit length in bumper-to-bumper conditions;
	\item $Q_A(\rho)$ and $Q_B(\rho)$ are functions of the density $\rho$ which modulate the variation of the velocity due to the noise term. In particular, we assume that the uncertainty in accelerating and braking vanishes in free road and jammed conditions, namely when vehicles tend to travel at the maximum speed or to stop, respectively. From now on, we will always take
	\begin{equation} \label{eq:QA-QB}
		Q_A(\rho)=P, \quad Q_B(\rho)=1-P.
	\end{equation}
	These simple choices allow one to simplify the microscopic model \eqref{eq:GenericMicroRules} and at the same time they reproduce realistic macroscopic behaviors of traffic flow, see Section \ref{sec:FundamentalDiagrams}.\\
	Let us illustrate the model considering a few typical cases. Let $\rho \approx \rhom$, namely let us suppose that the road is almost congested. Then $P(\rho) \approx 0$ and the second term in \eqref{eq:GenericMicroRulesAcc} vanishes. Since $Q_A(\rho) \approx 0$ when $\rho \approx \rhom$, it results $v \to \vb$ if $\vb < u$. Thus vehicles tend to keep their velocities even if they are slower than the speed of traffic as the road becomes congested. In contrast, if $\vb > u$, the post-interaction speed $v$ is less than the pre-interaction speed $\vb$, see \eqref{eq:GenericMicroRulesBrak}, thus fast vehicles reduce the velocity in order to reach the desired speed $V_B(u,\rho)$. On the contrary, if $\rho \approx 0$ then $P(\rho) \approx 1$, i.e. the probability of accelerating becomes larger as the road becomes free. Therefore, using the rules \eqref{eq:GenericMicroRules}, the post-interaction speed $v$ is greater than $\vb$ if $\vb < u$, namely slow vehicles increase their velocity in order to reach the desired speed $V_A(\vb,\rho)$. Conversely, $v \to \vb$ in the opposite case $\vb > u$, i.e. fast vehicles tend to keep their high speed.
\end{itemize}
Using definitions \eqref{eq:QA-QB}, the microscopic interaction rules \eqref{eq:GenericMicroRules} can be written as
\begin{subequations}
	\label{eq:MicroRules}
	\begin{align}
	v = \vb + \sqrt{\epsilon \, P(\rho)} \, \Dv_A(\vb,\rho) \left( \sqrt{\epsilon \, P(\rho)} + \xi \right), \quad &\text{if \; $\vb < u$} \label{eq:MicroRulesAcc}\\
	v = \vb - \sqrt{ \epsilon \, (1-P(\rho))} \, \Dv_B(\vb,u,\rho) \left( \sqrt{\epsilon \, (1-P(\rho))} + \xi \right), \quad &\text{if \; $\vb > u$} \label{eq:MicroRulesBrak}
	\end{align}
\end{subequations}
It is clear that the above expressions are defined by including two levels of stochasticity.
%The first one is modeled by the second terms at the right-hand side of equations \eqref{eq:MicroRules}, which make the jumps of velocity, both in acceleration and braking, function of the local density $\rho$ by means of the probability $P(\rho)$. Thus, they model the stochastic behavior of drivers in the sense that the output speed after an interaction is not deterministically fixed but it is prescribed by taking into account the level of congestion of the road. The second one is instead modeled by the third terms in \eqref{eq:MicroRules} which allow one to consider the inability of a driver to adjust precisely his speed, thereby modeling a stochastic perturbation due to uncertainty/unpredictability of the driver behavior.
The first one is modeled by the terms at the right-hand side of equations \eqref{eq:MicroRules} that do not involve the random variable $\xi$. These terms concern the non-deterministic outcome of the vehicle interactions and which give rise to an average post-interaction speed depending on the probability of accelerating. The second level of stochasticity is instead modeled by the terms in \eqref{eq:MicroRules} that involve the random variable $\xi$. These terms perturb the average post-interaction speed as a result of the unpredictability of the individual behavior of the drivers.

\begin{remark}[Bounded speeds]
	We expect that the post-interaction speed $v$ is such that $v\in[\vb,V_A]$ when accelerating (i.e. $\vb < u$) and $v\in[V_B,\vb]$ when braking (i.e. $\vb > u$). We can show that this is indeed guaranteed by making suitable assumptions on the random variable $\xi$. For the equation \eqref{eq:MicroRulesAcc} we can write
	\[
	v = \vb (1 - \epsilon P - \sqrt{\epsilon P} \xi) + V_A (\epsilon P + \sqrt{\epsilon P} \xi).
	\]
	Thus $v$ is a convex combination of the pre-interaction velocity $\vb$ and of the desired speed $V_A$ if and only if
	\[
	-\frac{\epsilon P}{\sqrt{\epsilon P}} \leq \xi \leq \frac{1 - \epsilon P}{\sqrt{\epsilon P}}.
	\]
	Instead, in the case \eqref{eq:MicroRulesBrak} we can write
	\[
	v \leq \vb (1 - \epsilon (1-P) + \sqrt{\epsilon (1-P)} \xi) + V_B (\epsilon (1-P) - \sqrt{\epsilon (1-P)} \xi).
	\]
	Thus $v$ results again in a convex combination of the pre-interaction velocity $\vb$ and of the desired speed $V_B$ if and only if
	\[
	\frac{\epsilon (1-P) - 1}{\sqrt{\epsilon (1-P)}} \leq \xi \leq \frac{\epsilon (1-P)}{\sqrt{\epsilon (1-P)}}. 
	\]
	Finally, the two bounds are simultaneously preserved if
	\[
	\max\left\{-\frac{\epsilon P}{\sqrt{\epsilon P}},\frac{\epsilon (1-P) - 1}{\sqrt{\epsilon (1-P)}}\right\} \leq \xi \leq \min\left\{\frac{1 - \epsilon P}{\sqrt{\epsilon P}},\frac{\epsilon (1-P)}{\sqrt{\epsilon (1-P)}}\right\}.
	\]
	Since in the sequel we will be interested in the behavior of the model for small values of $\epsilon$, then in the limit we can allow $\abs{\xi} \in [0,+\infty)$. However, observe that the post-interaction speeds are again bounded because $v \to \vb$ for $\epsilon \to 0$.\\
	The above bounds do not preserve in principle that $\xi$ has zero mean. Nevertheless, bounds satisfying this request can be easily obtained. For instance, in the simple case $\epsilon=1$, straightforward computations lead to the following bounds
	\begin{gather*}
	\abs{\xi} \leq \frac{1-P}{\sqrt{P}}, \quad \text{if $P \geq 1/2$} \quad \text{or} \quad
	\abs{\xi} \leq \frac{P}{\sqrt{1-P}}, \quad \text{if $P < 1/2$}
	\end{gather*}
	which ensure that $\xi$ has zero mean.
\end{remark}

\subsection{Comparison with other models}\label{sec:Comparison} The interaction rules prescribed in \eqref{eq:MicroRules} can be obtained as a generalization of the stochastic interaction rules given in \cite{PgSmTaVg2}. There the post-interaction speed is given by
\begin{numcases}{v=}
	V_A \text{ with probability  $P$}, \, \vb \text{ with probability  $1-P$}, \quad &\text{if \; $\vb < \va$} \label{eq:PSTV-MicroRulesAcc}\\
   	\vb \text{ with probability  $P$}, \, V_B \text{ with probability  $1-P$}, \quad &\text{if \; $\vb > \va$} \label{eq:PSTV-MicroRulesBrak}
\end{numcases}
where $\va$ is the microscopic velocity of the leading vehicle. In fact, in \cite{PgSmTaVg2}, we consider binary interactions in which a vehicle first compares its velocity $\vb$ with the velocity $\va$ of the vehicle ahead, and then decides in probability how to change its speed. Since in the present framework we are interested in mean field interactions, $\va$ is replaced by the macroscopic speed $u$ and, apart from the additional noise term $\xi$ and the scaling factor $\epsilon$, the rules \eqref{eq:MicroRules} can be indeed recovered as mean speed of the rules \eqref{eq:PSTV-MicroRulesAcc}-\eqref{eq:PSTV-MicroRulesBrak}.

Notice that \eqref{eq:PSTV-MicroRulesAcc}-\eqref{eq:PSTV-MicroRulesBrak} correspond to the so-called $\delta$ model in \cite{PgSmTaVg2}, which describes the situation of instantaneous change of velocity produced by jumps in velocity, from $\vb$ to $V_A$ (acceleration) or to $V_B$ (braking). Clearly, more refined models are present in the literature, see e.g. \cite{KlarWegener96}. Also in \cite{PgSmTaVg2} a less simplistic model is considered, in which the output speed is uniformly distributed over a bounded range of velocities, the so-called $\chi$ model. However, looking at the fundamental diagrams of traffic, it is proved that the essential information at equilibrium is already caught by the $\delta$ model.

This fact can be also mathematically investigated by computing the evolution equation for the macroscopic speed. In fact, assuming mean field interactions (in the sense specified above), the model introduced in \cite{PgSmTaVg2} specializes as
\begin{equation}\label{eq:PSTV-Model}
	\partial_t f(t,v) = \rho \int_0^{\vm} \Prob{u} f(t,\vb) \mathrm{d}\vb - \rho f(t,v)
\end{equation}
where $\Prob{u}$ is the transition probability modeling the stochastic microscopic interaction rules. It writes as
\[
	\Prob{u}=
			\begin{cases}
				P \Acc (v) + (1-P) \delta_{\vb}(v), & \text{if \; $\vb < u$}\\
				P \delta_{\vb}(v) + (1-P) \Brak (v), & \text{if \; $\vb > u$}
			\end{cases}
\]
where $\Acc$ and $\Brak$ are probability densities prescribing the speed after acceleration and braking, respectively. In the $\delta$ model (see \eqref{eq:PSTV-MicroRulesAcc}-\eqref{eq:PSTV-MicroRulesBrak}) it results
\begin{equation}\label{eq:DeltaModel}
	\Acc(v)=\delta_{V_A}(v), \quad \Brak(v)=\delta_{V_B}(v)
\end{equation}
while for the $\chi$ model one has
\begin{equation}\label{eq:ChiModel}
	\Acc(v)=\frac{\chi_{[\vb,V_A]}(v)}{V_A-\vb}, \quad \Brak(v)=\frac{\chi_{[V_B,\vb]}(v)}{\vb-V_B}.
\end{equation}
Substituting the explicit expression of $\Prob{u}$ in \eqref{eq:PSTV-Model}, multiplying for $v$ and integrating over the velocity space $\V$, we find the evolution equation for the macroscopic speed. In particular, for the choice \eqref{eq:DeltaModel} we get
\[
	\frac{\mathrm{d}}{\mathrm{d}t} u(t) = P \int_{0}^{u} (V_A-\vb) f(t,\vb) \mathrm{d}\vb + (1-P) \int_{0}^{u} (\vb-V_B) f(t,\vb) \mathrm{d}\vb
\]
while for the choice \eqref{eq:ChiModel} we obtain
\[
	\frac{\mathrm{d}}{\mathrm{d}t} u(t) = \frac12 \left\{ P \int_{0}^{u} (V_A-\vb) f(t,\vb) \mathrm{d}\vb + (1-P) \int_{0}^{u} (\vb-V_B) f(t,\vb) \mathrm{d}\vb \right\}.
\]
We notice that the two differential equations for $u$ differ only in a multiplicative constant. This means that the steady state of the macroscopic speed does not depend on the specific choice, either \eqref{eq:DeltaModel} or \eqref{eq:ChiModel}, of the probability densities $\Acc$ and $\Brak$. In fact only the relaxation rate towards the equilibrium is influenced by the presence of the multiplicative factor.

The previous considerations justify the choice of the simpler $\delta$ model as a starting point for the derivation of the interaction rules \eqref{eq:MicroRules}. However, notice that in the present paper we include also the noise term, accounting for output speeds which are not precisely determined.
%From a mathematical point of view a noise term is needed in order to obtain the diffusive part in the Fokker-Planck equation, since without noise we can recover only a Vlasov-type equation in the grazing collision limit, see e.g. \cite{HertyIllner08}. This technical motivation can be also justified from a modeling point of view. In fact, the grazing limit can be viewed as an inexact perception of the speed of the leading vehicle.

The microscopic interaction rules given in \eqref{eq:MicroRules} are very similar to those chosen in \cite{HertyPareschi10}. There the post-interaction speed are given by
\begin{subequations}
	\label{eq:HP-MicroRules}
	\begin{align}
	v = \vb + A \, (V_A-\vb) + \nu(\vb) \, (V_A-\vb)^\kappa \, \xi, \quad &\text{if \; $\vb < u$} \label{eq:HP-MicroRulesAcc}\\
	v = \vb - B \, (\vb-V_B) + \nu(\vb) \, (\vb-V_B)^\kappa \, \xi, \quad &\text{if \; $\vb > u$} \label{eq:HP-MicroRulesBrak}.
	\end{align}
\end{subequations}
where $\kappa \geq 1$ is an exponent that calibrates the dependence of the noise on the braking/acceleration dynamics and $0 \leq \nu(\vb) \leq 1$ is a function that vanishes at the extreme values of the velocity.

However, some important differences are introduced in the present work. In fact, in \eqref{eq:HP-MicroRules} the jumps of velocity modeled by the second terms of the right-hand side are fixed since they do not depend on the level of congestion of the traffic. Moreover, these jumps are modulated by two non-negative constants $A,B \in [0,1]$ controlling the strength of the interactions, which, however, do not play the role of the time-scaling parameters. Therefore, when a time scaling is applied, the authors need to make assumptions on the way in which the ratio of these parameters behaves in the grazing collision limit. In this way, they produce some arbitrary constants which are not actually model parameters and they use them in order to recover some Fokker-Planck model already studied in the literature.

\subsection{Derivation of the Boltzmann equation with the microscopic interaction rules \eqref{eq:MicroRules}} \label{sec:DerivationBoltzmann} As mentioned in the Introduction, in this paper we recover the Fokker-Planck equation as a limit of a Boltzmann model for traffic flow with the aim of making the study of stationary solutions and of fundamental diagrams amenable to analysis. To this end, it is useful to derive preliminarily a weak form of the Boltzmann equation.

The derivation of a general Boltzmann-type equation for the case of binary interactions whose post-interaction states are influenced by random variables is given e.g in \cite[Chap. 1, Sect. 1.4]{PareschiToscaniBOOK} and in \cite{HertyKlarPareschi} for traffic flow models. However, since in the present framework we deal with mean field interactions, we follow a different approach based on a formulation similar to \eqref{eq:PSTV-Model}.

For the post-interaction rules \eqref{eq:MicroRules} the transition probability is simply
\[
	\mathcal{P}(\vb\to w|u;\rho) = \delta_v(w),
\]
where $v$ is given case-wise by either \eqref{eq:MicroRulesAcc} or \eqref{eq:MicroRulesBrak} and includes also the random variable $\xi$.

The evolution equation for the kinetic distribution $f$ becomes then
\begin{equation}\label{eq:NoWFModel}
	\partial_t f(t,w) = \int_{-\infty}^{+\infty} \int_0^{\vm} \mathcal{P}(\vb\to w|u;\rho) f(t,\vb) \mathrm{d}\vb \eta(\xi) \mathrm{d}\xi - f(t,w).
\end{equation}
To pass to the weak form we multiply equation \eqref{eq:NoWFModel} by a test function $\phi\in C(\V)$ and integrate over $\V$:
\begin{equation}\label{eq:BoltzmannModel}
	\frac{\mathrm{d}}{\mathrm{d}t} \int_0^{\vm} \phi(w) f(t,w) \mathrm{d}w = \left\langle \int_0^{\vm} (\phi(v)-\phi(\vb)) f(t,\vb) \mathrm{d}\vb \right\rangle
\end{equation}
where the operator $\langle \cdot \rangle$ denotes the mean with respect to the distribution $\eta$, namely
\[
	\left\langle g \right\rangle = \int_{-\infty}^{+\infty} g(\xi) \eta(\xi) \mathrm{d}\xi.
\]

From \eqref{eq:BoltzmannModel} it is clear that $\phi(w)=1$ is a collision invariant which guarantees the conservation of the total number of vehicles. Moreover, it represents the only conservation property satisfied by the kinetic equation, as usual in a traffic flow model.

%\begin{remark}[Macroscopic speed] \label{rm:MacroSpeed}
%	If we consider $\phi(v)=v$ in \eqref{eq:BoltzmannModel} and we divide by the density $\rho$ we can recover the evolution equation for the macroscopic speed
%	\[
%		\frac{\mathrm{d}}{\mathrm{d}t} u(t) = \frac{1}{\rho} \langle \int_0^{\vm} (v-\vb) f(t,\vb) \mathrm{d}\vb \rangle.
%	\]
%	Using the microscopic rules \eqref{eq:MicroRules} given at the beginning of this section, we get
%	\begin{equation}\label{eq:EqMacroSpeed}
%		\frac{\mathrm{d}}{\mathrm{d}t} u(t) = \frac{1}{\rho} \int_0^u \epsilon P \Dv_A f(t,\vb) \mathrm{d}\vb - \frac{1}{\rho} \int_u^{\vm} \epsilon (1-P) \Dv_B f(t,\vb) \mathrm{d}\vb.
%	\end{equation}
%	Since $\Dv_A=V_A-\vb$ and $\Dv_B=\vb-V_B$, where $V_B$ is not a function of the pre-interaction speed $\vb$, the equation writes as
%	\[
%		\frac{\mathrm{d}}{\mathrm{d}t} u(t) = \epsilon \left( - P u - (1-2P) u^+ + \frac{1-P}{\rho} \rho^+ V_B + \frac{P}{\rho} \int_{0}^u V_A f(t,\vb) \mathrm{d}\vb \right),
%	\]
%	where $\rho^+$ and $u^+$ are the partial moments defined as
%	\[
%		\rho^+ = \int_u^{\vm} f(t,v) \mathrm{d}v, \quad u^+ = \frac{1}{\rho} \int_u^{\vm} v f(t,v) \mathrm{d}v.
%	\]
%	We can thus state that the macroscopic speed depends strongly only on the choice of the desired speeds $V_A$ and $V_B$. This fact will be observed also in Section \ref{sec:FundamentalDiagrams} when we will show that several diagrams of traffic can be obtained changing the modeling of the desired speeds.	
%\end{remark}

%\begin{remark}[Synchronized flow] \label{rm:Synchro}
\paragraph{Synchronized flow}
	If we consider $\phi(w)=w$ in \eqref{eq:BoltzmannModel} and we divide by the density $\rho$ we can recover the evolution equation for the macroscopic speed:
	\[
		\frac{\mathrm{d}}{\mathrm{d}t} u(t) = \frac{1}{\rho} \left\langle \int_0^{\vm} (v
		-\vb) f(t,\vb) \mathrm{d}\vb \right\rangle.
	\]
	Using the microscopic rules \eqref{eq:MicroRules} to express the post-interaction speed $v$, we get
	\begin{equation}\label{eq:EqMacroSpeed}
		\frac{\mathrm{d}}{\mathrm{d}t} u(t) = \frac{1}{\rho} \int_0^u \epsilon P \Dv_A f(t,\vb) \mathrm{d}\vb - \frac{1}{\rho} \int_u^{\vm} \epsilon (1-P) \Dv_B f(t,\vb) \mathrm{d}\vb,
	\end{equation}
	where we have used the fact that the random variable $\xi$ has zero mean. Suppose that the desired speeds $V_A$ and $V_B$ do not depend on the pre-interaction velocity $\vb$ and take, as assumed for instance in \cite{HertyPareschi10}, $V_A=V_B=u$. In this case $\Dv_A = -\Dv_B$ and using equation \eqref{eq:EqMacroSpeed} a straightforward computation provides
	\begin{equation} \label{eq:EvolutionSynchro}
		\frac{\mathrm{d}}{\mathrm{d}t} u(t) = -\frac{\epsilon}{2\rho} (1-2P) \int_0^{\vm} \abs{u-\vb} f(t,\vb) \mathrm{d}\vb.
	\end{equation}
	Observe that $\frac{\mathrm{d}}{\mathrm{d}t} u(t) = 0$ if and only if $f(v) = \rho \delta_u(v)$ which gives thus a steady solution corresponding to the phenomenon of synchronized traffic flow, in which all vehicles travel at the same speed $u$.\\
	Moreover, notice that $\frac{\mathrm{d}}{\mathrm{d}t} u(t) > 0$ if $P > 1/2$, which means that $u$ is increasing, i.e. $u \to \vm$, until $f(v) = \rho \delta_u(v)$. While, if $P < 1/2$, $\frac{\mathrm{d}}{\mathrm{d}t} u(t) < 0$, which means that $u$ is decreasing, i.e. $u \to 0$, until $f(v) = \rho \delta_u(v)$. Therefore the sign of $\frac{\mathrm{d}}{\mathrm{d}t} u(t)$ is defined by the density $\rho$ at initial time while the equilibrium value of the speed depends also on the initial condition $u(0)$. We stress that this analysis provides information only on the structure of the steady state of the kinetic distribution $f$ but the macroscopic equation \eqref{eq:EvolutionSynchro} for the average speed $u$ is not sufficient by itself to compute the equilibrium value of $u$.
	
	\begin{figure}
		\centering
		\includegraphics[width=0.49\textwidth]{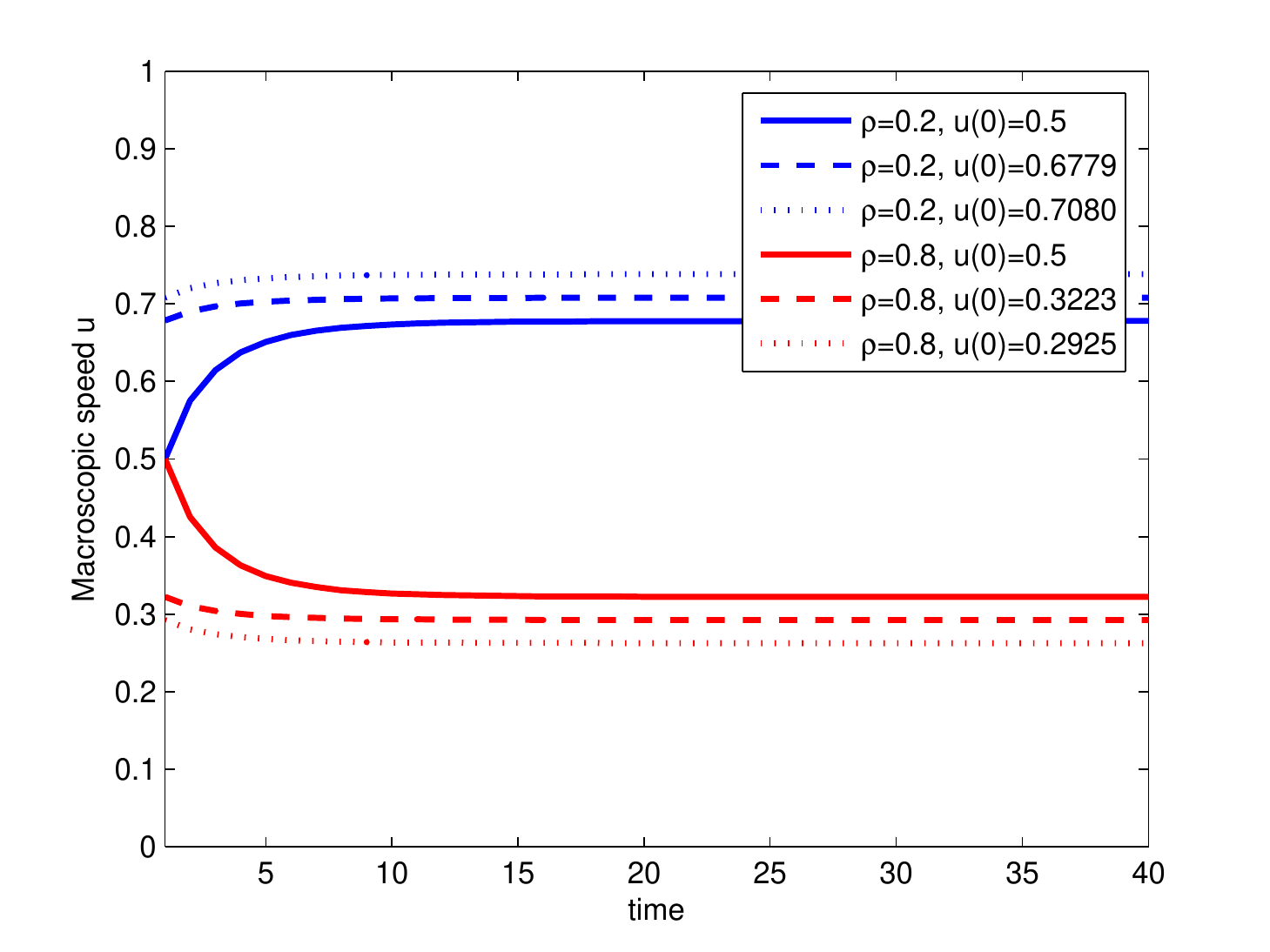}
		\includegraphics[width=0.49\textwidth]{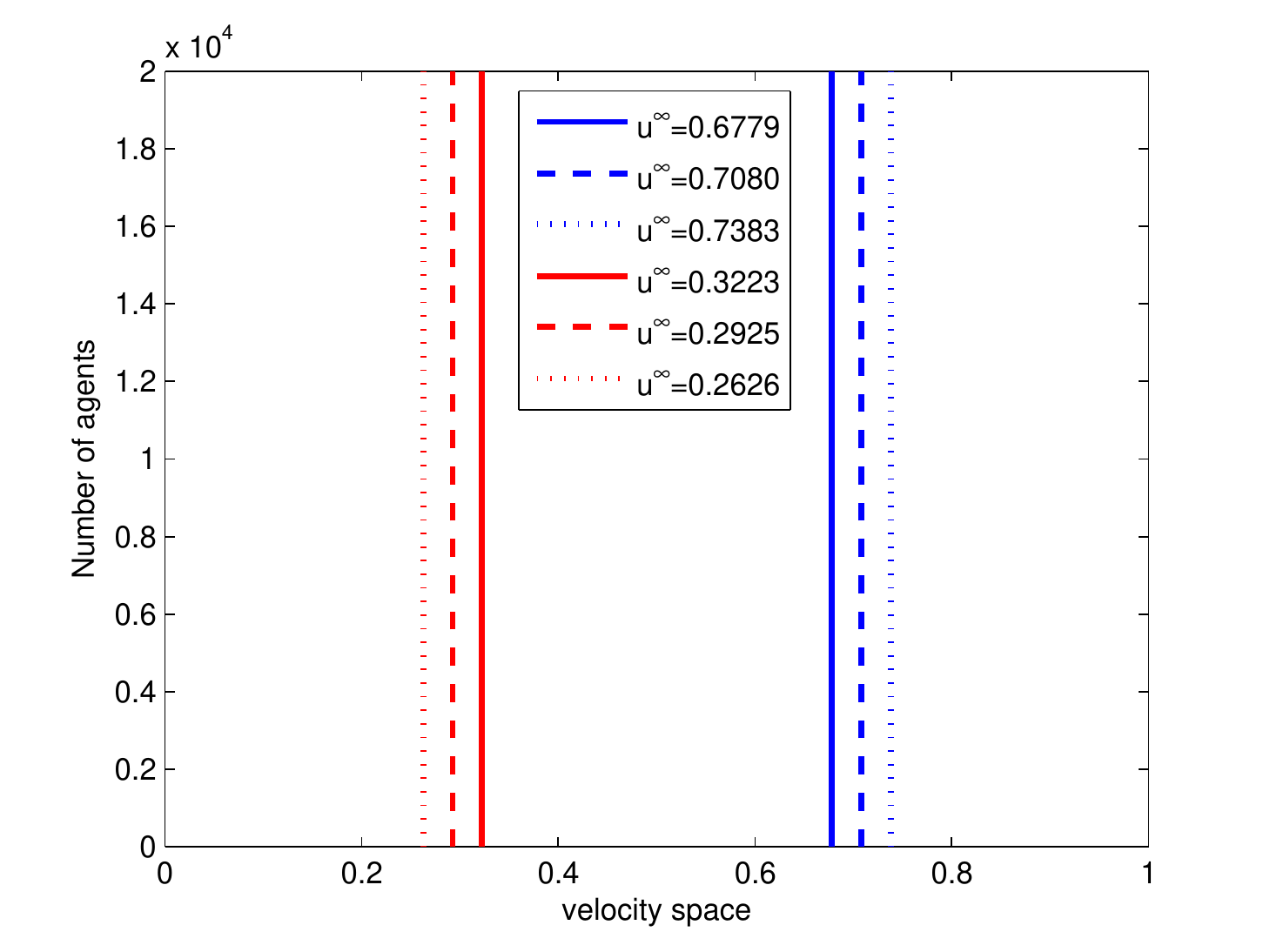}
		\caption{Left: evolution towards equilibrium of the macroscopic speed $u$ starting from different initial conditions. Right: distribution of the microscopic velocities at equilibrium.\label{fig:DSMC}}
	\end{figure}
	
	This aspect is also numerically investigated in Figure \ref{fig:DSMC} in which we propose a DSMC (Direct Simulation Monte Carlo) simulation for the kinetic equation \eqref{eq:BoltzmannModel} with the microscopic rules \eqref{eq:MicroRules} taking $V_A=V_B=u$. In particular, as in several kinetic models, we use the Nanbu-like asymptotic method \cite{Nanbu} which we reformulate in Algorithm \ref{alg:Nanbu} for the model \eqref{eq:BoltzmannModel}-\eqref{eq:MicroRules}. For further details see e.g. the paper \cite{AlbiPareschi} on microscopic models of flocking and swarming or the paper \cite{HertyKlarPareschi} on traffic flow models or the book \cite{PareschiToscaniBOOK}.
	
	\begin{algorithm}[t!]
		\caption{Nanbu algorithm for the model~\eqref{eq:BoltzmannModel}-\eqref{eq:MicroRules} with $V_A=V_B=u$ and $\epsilon=1$.}
		\label{alg:Nanbu}
		\begin{algorithmic}[1]
			\STATE Take $N$ samples of the microscopic velocities $v_j^0$, $j=1,\dots,N$ from the initial density $f^0(v)$;
			\STATE fix $\rho_0$ being the initial density of vehicles and thus the probability of changing velocity $P$;
			\FOR{$n=0$ \TO $M$}
				\STATE compute the macroscopic speed $u^n=\frac{1}{N} \sum_{j=1}^N v_j^i$;
				\FOR{$j=1$ \TO $N$}
					\IF{$P\geq 1/2$}
						\STATE sample $\xi$ from a zero mean distribution $\eta$ living on $\left[-\frac{1-P}{\sqrt{P}},\frac{1-P}{\sqrt{P}}\right]$;
					\ELSE
						\STATE sample $\xi$ from a zero mean distribution $\eta$ living on $\left[-\frac{P}{\sqrt{1-P}},\frac{P}{\sqrt{1-P}}\right]$;
					\ENDIF
					\IF{$v_j^n\leq u^n$}
						\STATE compute $v_j^{n+1}=v_j^n+\sqrt{P}(u-v_j^n)\left(\sqrt{P}+\xi\right)$;
					\ELSE
						\STATE compute $v_j^{n+1}=v_j^n-\sqrt{1-P}(v_j^n-u)\left(\sqrt{1-P}+\xi\right)$.
					\ENDIF
				\ENDFOR
			\ENDFOR
		\end{algorithmic}
	\end{algorithm}
	
	In Figure \ref{fig:DSMC}, Algorithm \ref{alg:Nanbu} is applied with $N=20000$ samples and $M=40$ time steps. In the left panel we show the evolution towards equilibrium of the macroscopic speed $u$ for two values of the density of vehicles. Precisely, the blue lines are referred to $\rho=0.2$, and thus $P=0.8$ using the definition \eqref{eq:Plaw}, while the red lines are referred to $\rho=0.8$, and thus $P=0.2$. We use different hatchings for the lines in order to mark the evolution obtained with different starting values of the mean speed. We notice that we find different equilibrium speeds depending on the initial condition $u(0)$ while the trend is influenced by the value of $P$. In fact, in each case, if $P > 1/2$ (resp. $P < 1/2$), the mean speed increases (resp. decreases), see equation \eqref{eq:EvolutionSynchro}, until it reaches the equilibrium value $u^\infty$ which is function of $u(0)$. Thus, as usual, the stationary mean speed can be computed by using the solution of the kinetic model while the macroscopic equation \eqref{eq:EvolutionSynchro} is not sufficient to determine $u^\infty$.
	
	For instance, focus on the case $P=0.8$. The evolution marked with the solid line is obtained by sampling the particles from a uniform distribution on $[0,1]$ at initial time. Thus the initial value of the mean speed $u$ is $u(0)\approx 0.5$. Since $P>1/2$, the speed $u$ increases during time evolution until the equilibrium value $u^\infty\approx0.6779$. This value is used as initial condition of the evolution marked with the dashed line. The trend is again increasing, since $P>1/2$, and we obtain a different equilibrium value $u^\infty\approx0.7080$. The same procedure is applied in order to compute the evolution marked with the dotted line and similar considerations hold for the case $P=0.2$.
		
	Finally, in the right panel of Figure \ref{fig:DSMC} we show that all the equilibrium values of the mean speed are indeed reached when all vehicles travel at the same velocity, i.e. $f(v)=\rho\delta_{u^\infty}(v)$ being $u^\infty$ function of $u(0)$.
	
	Although we can say nothing on the uniqueness of the synchronized steady solutions, we emphasize that this situation can occur at each density, while in experimental studies it occurs only for certain values of $\rho$, see e.g. \cite{kerner2004BOOK}. For this reason, we will not analyze deeply this aspect in the following but we will show that this solution is preserved by the Fokker-Planck approximation.
%\end{remark}

\section{Fokker-Planck approximation} \label{sec:FokkerPlanck}

The main drawback of a Boltzmann model is the complexity of the integral collision operator which makes the investigation of steady states difficult. Observe that the explicit knowledge of the asymptotic distribution is crucial both for computing the diagrams of traffic without expensive numerical simulations and for developing closure laws for macroscopic model. In fact, a kinetic model provides quite naturally speed-density relations arising from the microscopic interaction rules among the vehicles rather than from intuitive considerations.

In order to obtain simpler kinetic models, the goal is to replace the integral operator with differential operators. To this end, starting from the Boltzmann model \eqref{eq:BoltzmannModel} we derive a Fokker-Planck model using a technique similar to the one introduced for instance in \cite{Desvillettes,DiPernaLions,PareschiToscaniVillani,Villani1999} and already used in \cite{HertyPareschi10} for non-symmetric interaction rules, or in \cite{PareschiToscani2006} for gas-dynamics and in \cite{Toscani2006} for opinion formation. However, as discussed at the end of Section \ref{sec:Comparison}, here, in contrast with \cite{HertyPareschi10}, the interaction rules \eqref{eq:MicroRules} allow us to simplify the procedure of the grazing collision limit avoiding the definition of several free parameters and additional assumptions on the test function.

From now on, let us to consider $\phi\in C_c^{\infty}(\V)$. Starting from \eqref{eq:BoltzmannModel}, we first scale the time using the same parameter $\epsilon$ appearing in \eqref{eq:MicroRules}, setting
\[
	\tau=\epsilon t, \quad f(t,v)=\tilde{f}(\tau,v).
\]
For brevity, in the following we omit the use of the tilde and we indicate by $f$ the scaled distribution function. Then, braking up the integrals, equation \eqref{eq:BoltzmannModel} writes as
\[
	\partial_{\tau} \int_0^{\vm} \phi(w) f(\tau,w) \mathrm{d}w = \mathcal{A} + \mathcal{B}.
\]
where
\[
	\mathcal{A} = \frac{1}{\epsilon} \left\langle \int_0^u (\phi(v)-\phi(\vb)) f(\tau,\vb) \mathrm{d}\vb \right\rangle, \quad \mathcal{B} = \frac{1}{\epsilon} \left\langle \int_u^{\vm} (\phi(v)-\phi(\vb)) f(\tau,\vb) \mathrm{d}\vb \right\rangle.
\]
are the acceleration and braking collision kernels respectively. We study explicitly only the case $\vb < u$. Similar computations can be easily repeated for $\vb > u$. %Thus, consider the acceleration part of the collision kernel
%\begin{equation}\label{eq:CollisionOpAcc}
%	\frac{1}{\epsilon} \langle \int_0^{u} (\phi(v)-\phi(\vb)) f(\tau,\vb) \mathrm{d}\vb \rangle.
%\end{equation}
Since for small values of $\epsilon$ we have $v \approx \vb$, we can compute the Taylor expansion up to second order of $\phi(v)-\phi(\vb)$ around $\vb$:
\[
	\phi(v)-\phi(\vb) = \phi'(\vb)(v-\vb) + \frac12 \phi''(\overline{\vb})(v-\vb)^2, \quad \overline{\vb} \in (\vb,v)
\]
with $v$ defined in this case by \eqref{eq:MicroRulesAcc}. Substituting $v-\vb$ in the above expansion, the acceleration part of the collision kernel becomes 
\begin{align*}
	\mathcal{A} = & \frac{1}{\epsilon} \left\langle \int_0^u (\epsilon P \Dv_A + \sqrt{\epsilon P} \Dv_A \xi) \phi'(\vb) f(\tau,\vb) \mathrm{d}\vb \right\rangle \\
	& + \frac{1}{2\epsilon} \left\langle \int_0^u (\epsilon P \Dv_A + \sqrt{\epsilon P} \Dv_A \xi)^2 \phi''(\vb) f(\tau,\vb) \mathrm{d}\vb \right\rangle \\
	& + \frac{1}{\epsilon} \langle R_{\mathcal{A}}(f,\eta,\phi) \rangle,
\end{align*}
where
\[
	R_{\mathcal{A}}(f,\eta,\phi)=\frac12 \int_0^{u} (\epsilon P \Dv_A + \sqrt{\epsilon P} \Dv_A \xi)^2 (\phi''(\overline{\vb})-\phi''(\vb)) f(\tau,\vb) \mathrm{d}\vb
\]
is the remaining term that vanishes as $\epsilon \to 0$. In fact, since $\phi$ is a smooth function, the second derivative $\phi''$ is Lipschitz continuous, thus
\[
\exists \, \mathrm{L} \geq 0 \text{ s.t. } \abs{\phi''(\overline{\vb})-\phi''(\vb)} \leq \mathrm{L} \abs{\overline{\vb}-\vb}.
\]
Recalling that $\overline{\vb} \in (\vb,v)$, we have $\abs{\overline{\vb}-\vb}=\theta\abs{v-\vb}$ for some $\theta\in(0,1)$. Therefore
\[
\mathrm{L} \abs{\overline{\vb}-\vb} \leq \mathrm{L} \abs{v-\vb} \leq \mathrm{L} \mathrm{C} \sqrt{\epsilon}
\]
where $\mathrm{C}=\mathrm{C}(\max\{\Dv_A\},\xi_{\max})$ and
\begin{align*}
\frac{1}{\epsilon} \abs{\langle R_{\mathcal{A}}(f,\eta,\phi) \rangle} & \leq \frac{1}{2\epsilon} \int_{-\infty}^{+\infty} \int_0^u (v-\vb)^2 \abs{\phi''(\overline{\vb})-\phi''(\vb)} f(\tau,\vb) \mathrm{d}\vb \eta(\xi) \mathrm{d}\xi \\
& \leq \frac{\mathrm{L} \mathrm{C} \epsilon^{\frac32}}{2\epsilon} \int_{-\infty}^{+\infty} \int_0^u f(\tau,\vb) \mathrm{d}\vb \eta(\xi) \mathrm{d}\xi \xrightarrow{\epsilon\to 0^+} 0.
\end{align*}
Thus, computing the limit for $\epsilon \to 0$, corresponding to a situation in which the interactions are frequent but they produce only small variations in the output velocities, and using the identities
\[
	\int_{-\infty}^{+\infty} \eta(\xi) \mathrm{d}\xi = 1, \quad \int_{-\infty}^{+\infty} \xi \eta(\xi) \mathrm{d}\xi = 0, \quad \int_{-\infty}^{+\infty} \xi^2 \eta(\xi) \mathrm{d}\xi = \sigma^2,
\]
we obtain
\[
	\lim_{\epsilon \to 0} \mathcal{A} = P \int_0^u \Dv_A \phi'(\vb) f(\tau,\vb) \mathrm{d}\vb + \frac{\sigma^2 P}{2} \int_0^u \left(\Dv_A\right)^2 \phi''(\vb) f(\tau,\vb) \mathrm{d}\vb.
\]

The same considerations apply also to the braking part of the collision kernel. In particular, one verifies that $R_{\mathcal{B}}(f,\eta,\phi) \to 0$ as $\epsilon \to 0$ and
\begin{align*}
	\lim_{\epsilon \to 0} \mathcal{B} = & -(1-P) \int_u^{\vm} \Dv_B \phi'(\vb) f(\tau,\vb) \mathrm{d}\vb + \frac{\sigma^2 (1-P)}{2} \int_u^{\vm} \left(\Dv_B\right)^2 \phi''(\vb) f(\tau,\vb) \mathrm{d}\vb.
\end{align*}

Finally, in order to derive the strong formulation of the Fokker-Planck equation, we write the weak form resulting from the grazing collision limit as
\begin{equation}\label{eq:WeakFP}
\begin{aligned}
	\partial_{\tau} \int_0^{\vm} \phi(v) f(\tau,v) \mathrm{d}v =& P \int_0^{\vm} \chi_{[0,u]}(\vb) \phi'(\vb) \Dv_A f(\tau,\vb) \mathrm{d}\vb \\
	& - (1-P) \int_0^{\vm} \chi_{[u,\vm]}(\vb) \phi'(\vb) \Dv_B f(\tau,\vb) \mathrm{d}\vb\\
	& + \frac{\sigma^2 P}{2} \int_0^{\vm} \chi_{[0,u]}(\vb) \phi''(\vb) \left(\Dv_A\right)^2 f(\tau,\vb) \mathrm{d}\vb \\
	& + \frac{\sigma^2 (1-P)}{2} \int_0^{\vm} \chi_{[u,\vm]}(\vb) \phi''(\vb) \left(\Dv_B\right)^2 f(\tau,\vb) \mathrm{d}\vb.
\end{aligned}
\end{equation}

Integrating by parts each term of the right-hand side of \eqref{eq:WeakFP} and using the fact that $\phi(0)=\phi(\vm)=\phi'(0)=\phi'(\vm)=0$, we recover the following strong form of the Fokker-Planck equation
\begin{equation}\label{eq:StrongFP}
\begin{aligned}
	\partial_{\tau} f(\tau,v) = & - \partial_v \left[ f(\tau,v) \left( P \chi_{[0,u]}(v) \Dv_A - (1-P) \chi_{[u,\vm]}(v) \Dv_B \right) \right] \\
	& + \frac{\sigma^2}{2} \partial_{vv} \left[ f(\tau,v) \left( P \chi_{[0,u]}(v) \left(\Dv_A\right)^2 + (1-P) \chi_{[u,\vm]}(v) \left(\Dv_B\right)^2 \right) \right]
\end{aligned}
\end{equation}
as limit of the Boltzmann model \eqref{eq:BoltzmannModel} based on the microscopic rules \eqref{eq:MicroRules}. Observe that \eqref{eq:StrongFP} can be rewritten in the usual formulation of a Fokker-Planck-type model for traffic flow:
\begin{equation}\label{eq:GeneralFP}
	\partial_{\tau} f(\tau,v) + \partial_v \left[ f(\tau,v) B(v,u,\rho) - D(v,u,\rho) \partial_v f(\tau,v) \right] = 0,
\end{equation}
see for instance the prototype example introduced in \cite{IllnerKlarMaterne}. In the present framework, the acceleration/braking operator $B(v,u,\rho)$ is
\begin{equation} \label{eq:AccBrakOp}
\begin{aligned}
	B(v,u,\rho) = & P \chi_{[0,u]}(v) \Dv_A - (1-P) \chi_{[u,\vm]}(v) \Dv_B \\
	& - \frac{\sigma^2}{2} \partial_v \left( P \chi_{[0,u]}(v) \left(\Dv_A\right)^2 + (1-P) \chi_{[u,\vm]}(v) \left(\Dv_B\right)^2 \right)
\end{aligned}
\end{equation}
while the diffusive operator $D(v,u,\rho)$ is
\begin{equation} \label{eq:DiffusiveOp}
	D(v,u,\rho) = \frac{\sigma^2}{2} \left( P \chi_{[0,u]}(v) \left(\Dv_A\right)^2 + (1-P) \chi_{[u,\vm]}(v) \left(\Dv_B\right)^2 \right).
\end{equation}
As we will see in Section \ref{sec:FundamentalDiagrams}, the stationary solution of \eqref{eq:GeneralFP} can be obtained avoiding the explicit computation of derivatives of characteristic functions that may result in Dirac functions. 

\begin{remark}[Boundary conditions]
	In contrast to the Boltzmann model, in which $\rho$ is automatically conserved, in the Fokker-Planck model the density remains constant in time and thus it is an invariant of the model \eqref{eq:GeneralFP} if, for any choices of $V_A$ and $V_B$, the following boundary condition holds
	\begin{equation}\label{eq:BoundCond}
		f(\tau,v) B(v,u,\rho) - D(v,u,\rho) \partial_v f(\tau,v)\Big\rvert_{v=0}^{v=\vm} = 0
	\end{equation}
	at each time $\tau$.
\end{remark}

\section{Fundamental diagrams} \label{sec:FundamentalDiagrams}

In this section, we investigate the fundamental diagrams of traffic flow resulting from the Fokker-Planck model \eqref{eq:GeneralFP}. Since we are interested in reproducing the features of experimental diagrams, as the phase transition and the multivalued behavior, we show that a particular choice of the desired speeds allows us to obtain diagrams which exhibit the qualitative properties of data.

As usual, macroscopic diagrams are recovered by computing the macroscopic quantities as moments of the time-asymptotic kinetic distribution $f^\infty(v)$. Notice that, using the Fokker-Planck model, $f^\infty(v)$ can be computed easily by solving an ordinary differential equation. In fact, at the steady state, the time derivative of the distribution function in \eqref{eq:GeneralFP} must be zero and thus we are led to the following homogeneous ODE:
\begin{equation} \label{eq:ODE}
	f'(v) = \frac{B(v,u,\rho)}{D(v,u,\rho)} f(v), \quad v\in[0,\vm]
\end{equation}
whose solution can be written in form of exponentials once the cases $v < u$ and $v > u$ are distinguished.

Using the separation of variables method, the time-asymptotic solution for general jumps of velocity $\Dv_A$ and $\Dv_B$ can be computed explicitly substituting \eqref{eq:AccBrakOp} and \eqref{eq:DiffusiveOp} in \eqref{eq:ODE} and it is given by
\begin{subequations}
	\label{eq:SteadySolution}
	\begin{align}
	f^\infty(v)=f^\infty(u^-) \left( \frac{\Dv_A(v,\rho)|_{v=u}}{\Dv_A(v,\rho)} \right)^2 \exp\left(-\frac{2}{\sigma^2} \int_v^u \frac{1}{\Dv_A(s,\rho)} \mathrm{d}s\right), \quad & \text{if \; $v < u$} \label{eq:SteadySolutionAcc}\\
	f^\infty(v)=f^\infty(u^+) \left( \frac{\Dv_B(v,u,\rho)|_{v=u}}{\Dv_B(v,u,\rho)} \right)^2 \exp\left(-\frac{2}{\sigma^2} \int_u^v \frac{1}{\Dv_B(s,u,\rho)} \mathrm{d}s\right), \quad & \text{if \; $v > u$} \label{eq:SteadySolutionBrak}
	\end{align}
\end{subequations}
which indeed depends on three unknown parameters, the mean speed $u$, $f^\infty(u^-)$ and $f^\infty(u^+)$. The latter two are the left and right limits, respectively, of $f^\infty$ at $v=u$, that here we take as two integration constants.

Equivalently we can think of \eqref{eq:SteadySolutionAcc} as parametrized by $r f^\infty(u^+)$, with
\begin{equation} \label{eq:r}
	r=\frac{f^\infty(u^-)}{f^\infty(u^+)}.
\end{equation}
Notice that imposing the constraint on the zeroth moment (i.e. mass conservation) of the distribution function
\begin{equation}\label{eq:ZeroConstraint}
	\rho = \int_0^{\vm} f^\infty(v) \mathrm{d}v
\end{equation}
is not sufficient to define uniquely both $r$ and $f^\infty(u^+)$. In fact, the constraint \eqref{eq:ZeroConstraint} leads to a single equation
\[
	f^\infty(u^+) \left( r \rho^- + \rho^+ \right) = \rho
\]
with two unknowns, where
\[
	\rho^- = \int_0^u f^\infty(v) \mathrm{d} v, \quad \rho^+ = \int_u^{\vm} f^\infty(v) \mathrm{d} v
\]
are the partial densities of slow and fast vehicles (with respect to the mean speed). Thus, we obtain a family of steady solutions and without loss of generality we assume that $f^\infty$ will be parametrized by the ratio $r$, so that
\begin{equation} \label{eq:IntConst}
	f^\infty(u^-) = \frac{r\rho}{r \rho^- + \rho^+}, \quad f^\infty(u^+) = \frac{\rho}{r \rho^- + \rho^+}.
\end{equation}
%We find them by imposing the constraint on the zeroth moment of the distribution function
%\begin{equation}\label{eq:ZeroConstraint}
%	\rho = \int_0^{\vm} f^\infty(v) \mathrm{d}v.
%\end{equation}
%Actually, notice that equation \eqref{eq:ZeroConstraint} is not sufficient to define uniquely the two constants of integration and thus it will lead to a family of steady solutions. Without loss of generality we impose that $f^\infty$ is parametrized by the ratio
%\begin{equation} \label{eq:r}
%	r=\frac{f^\infty(u^-)}{f^\infty(u^+)}.
%\end{equation}

For $r=1$ we obtain continuous asymptotic distributions at $v=u$, while for $r \neq 1$ we allow for discontinuous distributions at $v=u$. This is the mathematical key which allows us to recover multivalued diagrams of traffic as a result of the existence of a one-parameter family of solutions. In fact, for each value of the ratio $r$ we obtain a corresponding time-asymptotic solution. Thus different values of $r$ give different equilibrium speeds, for a fixed density $\rho$. Usually, in the literature only continuous solutions were considered, see e.g. \cite{IllnerKlarMaterne} in which a \textquotedblleft small\textquotedblright \ multivalued behavior is obtained by means of the probability of lane changing. However, experimental fundamental diagrams typically are discontinuous around a critical density which defines a sort of phase change. Thus it is quite natural to expect discontinuous asymptotic distributions.

The equilibrium value of the macroscopic speed $u$ is obtained by solving a non-linear equation coming from the constraint on the first moment of the distribution function and from \eqref{eq:ZeroConstraint}
\begin{equation}\label{eq:OneConstraint}
	\int_0^{\vm} v f^\infty(v) \mathrm{d}v = \rho u = u \int_0^{\vm} f^\infty(v) \mathrm{d}v.
\end{equation}
The identity \eqref{eq:OneConstraint} provides the nonlinear equation $R(u)=0$, with
\begin{equation} \label{eq:R}
	R(u):=\int_0^{\vm} (u-v) f^\infty(v) \mathrm{d}v.
\end{equation}
Solving it for each value of the density in $[0,\rho_{\max}]$ we find the speed at equilibrium which defines therefore a relation between the density $\rho$ and the mean speed $u$ of the flow. Since $R(u)=0$ defines a nonlinear equation, in principle it might have zero or more than one solution.

Finally, observe that without diffusion, namely if $\sigma^2 \to 0$, then the two exponentials in \eqref{eq:SteadySolution} tend to $1$. This means that the general steady state \eqref{eq:SteadySolution}, and thus the equilibrium speed, is influenced only by the choice of the desired speed $V_A$ and $V_B$. In fact, in this case the microscopic rules \eqref{eq:MicroRules} become purely deterministic, provided $P(\rho)$.

\begin{remark}[Synchronized flow]
	In Paragraph \textquotedblleft Synchronized flow\textquotedblright \ in Section \ref{sec:DerivationBoltzmann} we have proved that the distribution function $f(v) = \rho \delta_u(v)$ is a stationary solution of the Boltzmann model \eqref{eq:BoltzmannModel} based on the interactions \eqref{eq:MicroRules}, when $V_A=V_B=u$. This solution is preserved in the Fokker-Planck approximation. We use the weak form of the model \eqref{eq:GeneralFP} in order to prove that $f(v) = \rho \delta_u(v)$ is, in distributional sense, the trivial steady state of the Fokker-Planck-type equation. In fact, assuming $\phi$ as a test function with compact support in $\V$, $f(v)$ is a weak stationary solution if
	\begin{gather*}
		\int_0^{\vm} \phi(v) \partial_v \left( f(v) B(v,u,\rho) - D(v,u,\rho) \partial_v f(v) \right) \mathrm{d}v = 0.
	\end{gather*}
	Integrating by parts the left-hand side of the above equation we get
	\[
		- \int_0^{\vm} \phi'(v) B(v,u,\rho) f(v) \mathrm{d}v - \int_{0}^{\vm} \partial_v ( \phi'(v) D(v,u,\rho) ) f(v) \mathrm{d}v.
	\]
	Taking $f(v) = \rho \delta_u(v)$ we obtain
	\begin{gather*}
		\int_0^{\vm} \phi'(v) B(v,u,\rho) f(v) \mathrm{d}v = \rho B(u,u,\rho) = 0 \\
		\int_{0}^{\vm} \partial_v ( \phi'(v) D(v,u,\rho) ) f(v) \mathrm{d}v = \rho \phi''(v) D(u,u,\rho) + \rho \phi'(v) \partial_v D(v,u,\rho)|_{v=u} = 0.
	\end{gather*}
	In fact, the choice $V_A = V_B = u$ leads to a degeneracy at $v = u$ in the acceleration/braking operator $B$, in the diffusion operator $D$ and in its derivative, since:
	\begin{align*}
		B(v,u,\rho) = & P (u-v) \chi_{[0,u]}(v) - (1-P) (v-u) \chi_{[u,\vm]}(v) - \partial_v D(v,u,\rho)\\
		\partial_v D(v,u,\rho) = & \frac{\sigma^2}{2} \left[ -P (u-v)^2 \delta_u(v) - 2P (u-v) \chi_{[0,u]}(v) \right.\\
		& \left. + (1-P) (v-u)^2 \delta_u(v) + 2(1-P) (v-u) \chi_{[u,\vm]}(v) \right].
	\end{align*}
\end{remark}

Now, let us study two particular models obtained with different choices of the desired speeds. In both cases, we formulate the explicit steady state and the expression of the function $R(u)$ which defines the equilibrium speed. We show that the two choices provide meaningful diagrams of traffic (flux-density and speed-density relationships).

\subsection{Case 1}\label{sec:Case1} Here we assume
	\begin{equation} \label{eq:DesiredSpeeds1}
		V_A=v+P(\vm-v), \quad V_B=Pu.
	\end{equation}
	This means that, in acceleration, the desired speed is a certain velocity in $[v,\vm]$ depending on the value of the probability of accelerating $P$ and thus on the density of vehicles. Instead, when braking, the desired speed is a velocity in $[0,u]$ and depends only on macroscopic quantities, as discussed in Section \ref{sec:Boltzmann}. With this choice and evaluating explicitly the integrals appearing in \eqref{eq:SteadySolution}, the asymptotic distribution $f^\infty$ becomes
	\begin{subequations}
		\label{eq:Case1}
		\begin{align}
			f^\infty(v) = f^\infty(u^-) \left( \frac{\vm-u}{\vm-v} \right)^{c^A}, \quad & \text{if \; $v < u$} \label{eq:Case1Acc}\\
			f^\infty(v) = f^\infty(u^+) \left( \frac{u-Pu}{v-Pu} \right)^{c^B}, \quad & \text{if \; $v > u$} \label{eq:Case1Brak} 
		\end{align}
	\end{subequations}
	where $c^A=\frac{2}{\sigma^2 P}+2$ and $c^B=\frac{2}{\sigma^2}+2$. In the sequel we consider $P=1-\rho$. Recall that the equilibrium speed is found by solving the constraint \eqref{eq:OneConstraint}, which leads to the non-linear equation $R(u)=0$. For the steady solution \eqref{eq:Case1} we have
	\begin{align*}
		R(u) = & f^\infty(u^-) \underbrace{\int_0^u (u-v) \left( \frac{\vm-u}{\vm-v} \right)^{c^A} \mathrm{d}v}_{R_A(u)} - f^\infty(u^+) \underbrace{\int_u^{\vm} (v-u) \left( \frac{u-Pu}{v-Pu} \right)^{c^B} \mathrm{d}v}_{R_B(u)}.
	\end{align*}
	Computing explicitly both integrals we obtain
	\begin{subequations}\label{eq:Case1R}
	\begin{align}
		R_A(u)=&\frac{(\vm-u)^2}{(c^A-2)(c^A-1)}-\frac{\vm^{2-c^A}(\vm-u)^{c^A}}{(c^A-2)(c^A-1)}-\frac{u\vm^{1-c^A}(\vm-u)^{c^A}}{c^A-1} \label{eq:Case1RA}\\
		R_B(u)=&\frac{(u-Pu)^2}{(c^B-2)(c^B-1)}-\frac{(u-Pu)^{c^B}(\vm-Pu)^{2-c^B}}{(c^B-2)(c^B-1)} \label{eq:Case1RB} \\
		& -\frac{(u-Pu)^{c^B}(\vm-u)(\vm-Pu)^{1-c^B}}{(c^B-1)}
	\end{align}
	\end{subequations}
	Unfortunately, since $R(0)=R(1)=0$ we cannot ensure that $\exists u\in [0,\vm]$ such that $R(u)=0$. However, numerically we observe that the nonlinear equation $R(u)=0$ does not provide more than one solution $u\in(0,1)$.
	
	\begin{figure}[t!]
		%\centering
		%\includegraphics[width=0.52\linewidth]{Case1FD-sigma}
		\includegraphics[width=0.52\textwidth]{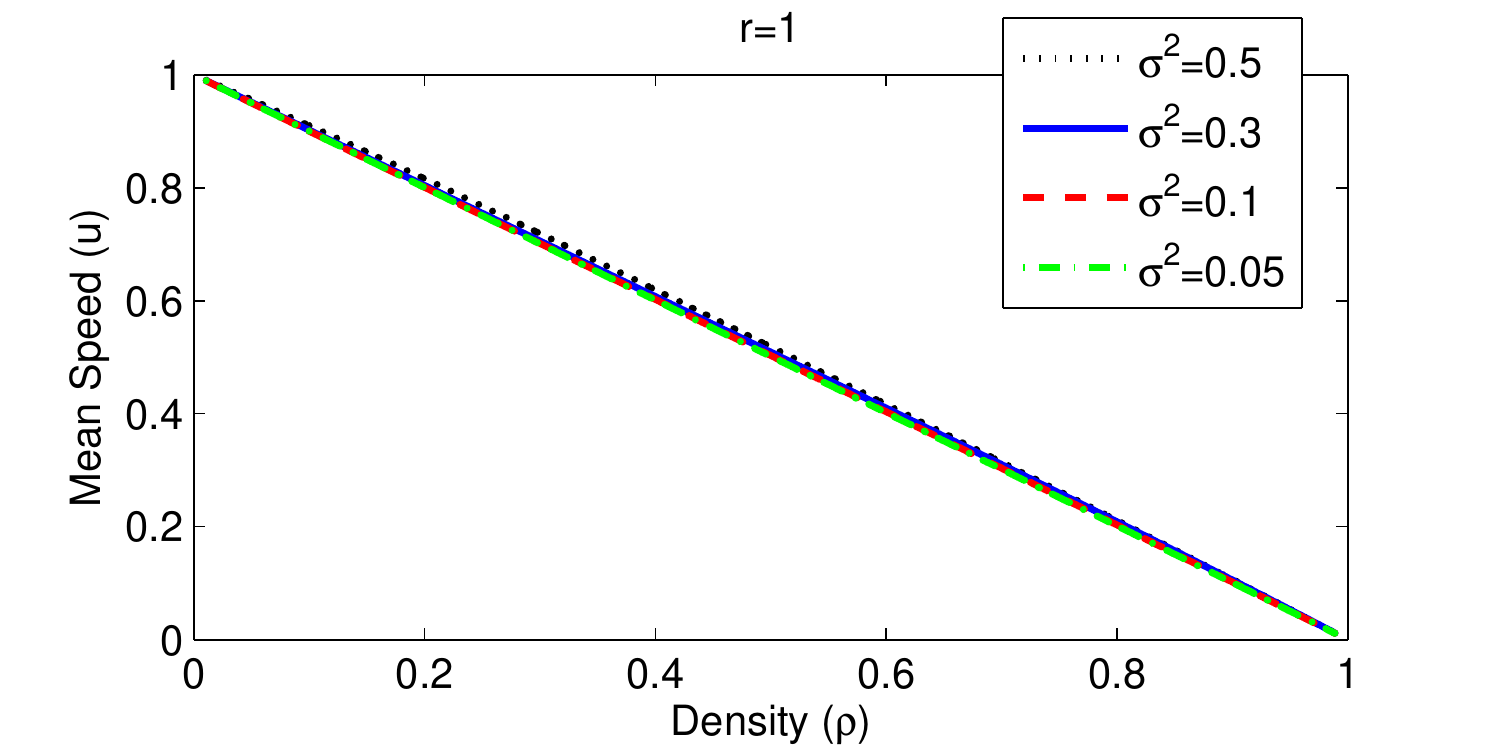}
		\includegraphics[width=0.52\textwidth]{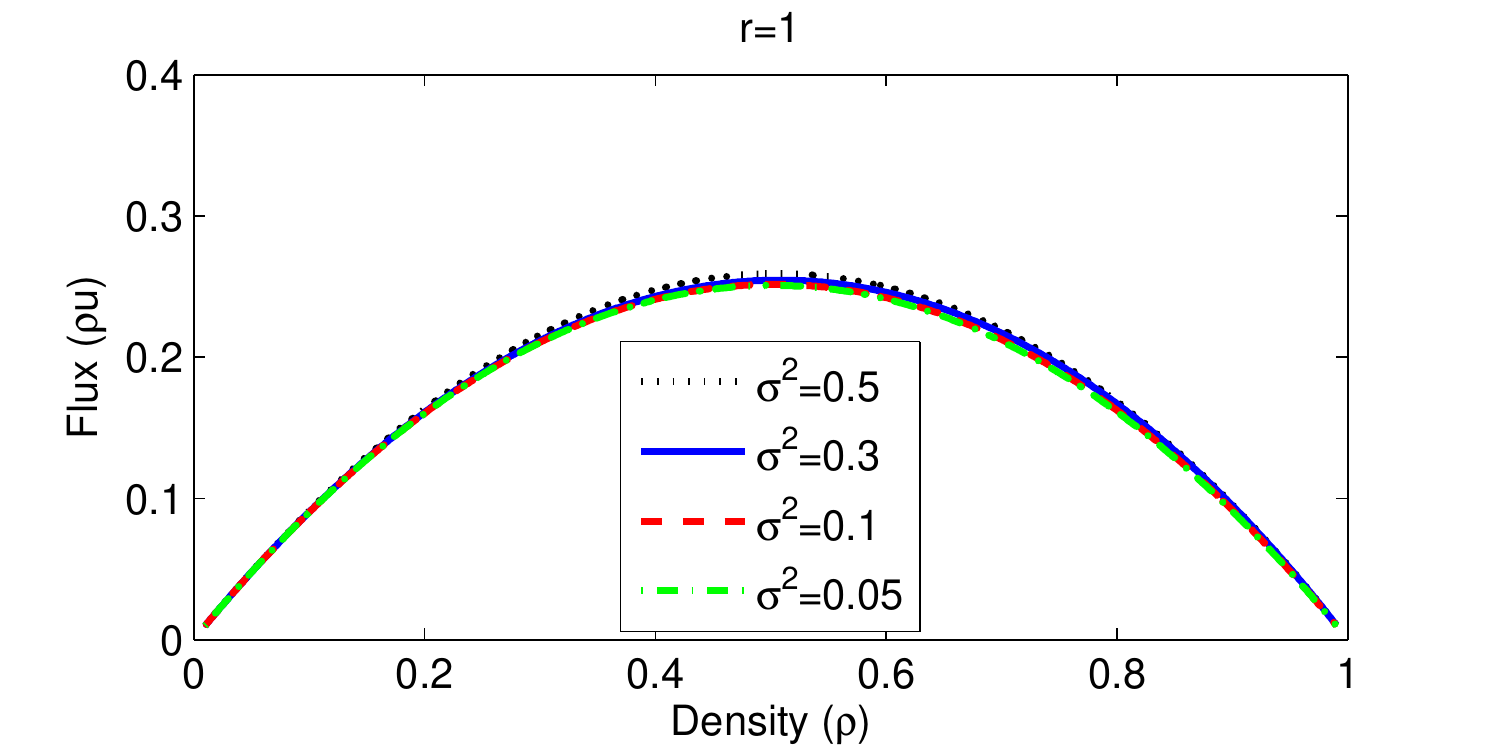}
		\caption{Diagrams of traffic using the desired speeds \eqref{eq:DesiredSpeeds1}. We take $r=1$ and we show the dependence on the variance $\sigma^2$.\label{fig:Case1FD-sigma}}
	\end{figure}
	
	\begin{figure}[t!]
		%\centering
		\includegraphics[width=0.52\textwidth]{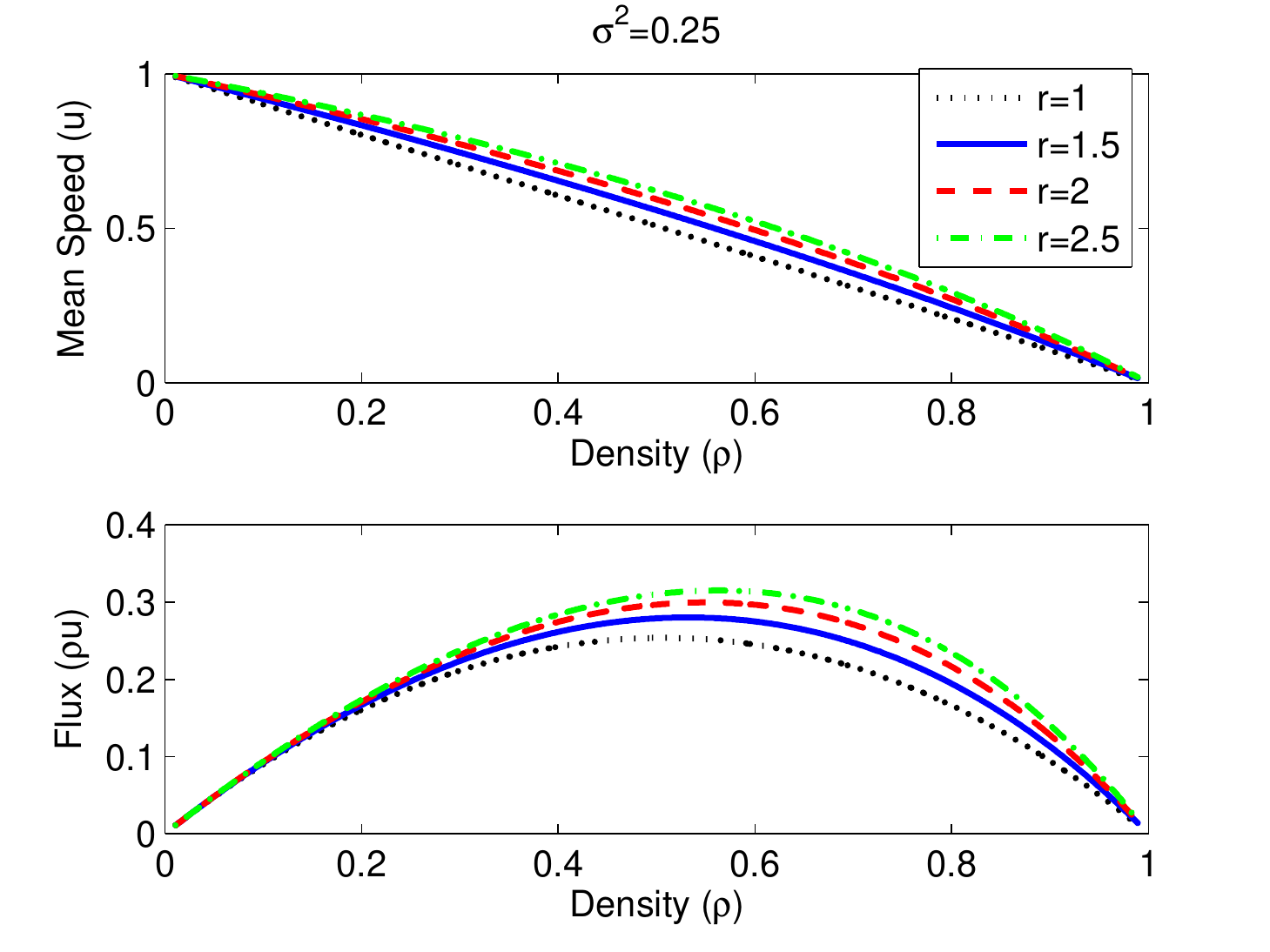}
		%\hfill
		\includegraphics[width=0.52\textwidth]{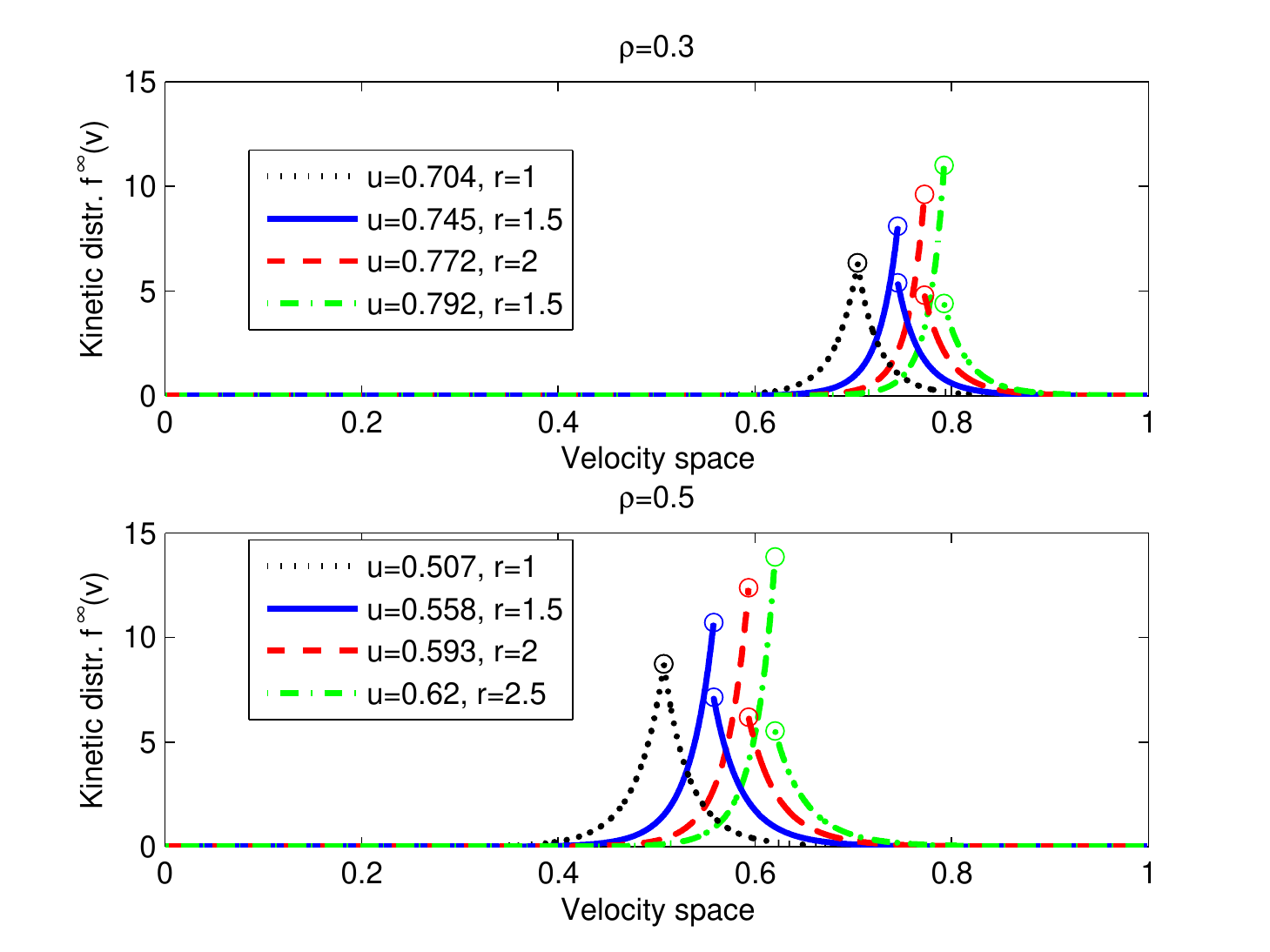}
		\caption{Left panels: diagrams of traffic using the desired speeds \eqref{eq:DesiredSpeeds1}. We take $\sigma^2=0.25$ and we study the dependence on the parameter $r$ which permits to find a region of scattered values. Right panels: we plot the related equilibrium distributions for $\rho=0.3$ (top) and $\rho=0.5$ (bottom). The circles define the values $f(v)$ for $v\to u^-$ and $v\to u^+$.\label{fig:Case1FD-r}}
	\end{figure}

	In Figure~\ref{fig:Case1FD-sigma} and~\ref{fig:Case1FD-r} we show the diagrams of traffic by varying the parameters of the model. In particular, in Figure~\ref{fig:Case1FD-sigma} we study the influence of %$\kappa$ (left panels) and $\sigma^2$ (right panels)
	$\sigma^2$ on the speed and on the flux. The diagrams are reproduced by considering only continuous steady state, namely we fix $r=1$ and then we compute $f^\infty(u^+)$ and $f^\infty(u^-)$ from equation \eqref{eq:ZeroConstraint} and \eqref{eq:r}. The equilibrium values are obtained by solving numerically $R(u)=0$ and the top plots show the speed-density diagrams, while the bottom ones show the flux-density diagrams. In all simulations we set $\rhom = \vm = 1$. The probability of accelerating $P$ is taken as prescribed in \eqref{eq:Plaw}.
	%The parameter $\kappa$ seems to control the symmetry of the fundamental diagram. In fact, for increasing $\kappa$ we observe that the value of the density in which there is the transition from the free phase and the congested phase of traffic tends to zero. Moreover, as $\kappa$ increases the absolute value of the slope of the flux curves at high densities becomes smaller and thus we have a smoother evolution towards the jammed traffic flow. Instead, the
	The variance $\sigma^2$ seems to produce only small variations in the equilibrium values of the speed. More precisely, it does not influence deeply the free and the congested phase of traffic but only the values near the phase transition.	Therefore, acting on %$\kappa$ and
	$\sigma^2$ we are not able to reproduce a meaningful scattering of the flux values.
	
	In the left panels of Figure \ref{fig:Case1FD-r} we consider the diagrams obtained with %$\kappa=2$ and
	$\sigma^2=0.25$, but modifying the parameter $r$, see equation \eqref{eq:r}. This allows us to take into account different steady state according to the jump at $v=u$. The corresponding stationary solutions appear in the right panels for two values of $\rho=0.3, 0.5$. We point out that now we obtain a more significant and accentuated dispersion of the values of macroscopic quantities in a large range of density. Note that the dispersion is very small for low and high densities, which is in accord with experimental surveys. The circles in the right plots are in order to show the values of the steady states $f^\infty(v,r)$ for $v\to u^-$ and $v\to u^+$, which coincide only for the black dotted distribution corresponding to the continuous case, with $r=1$. The role and the meaning of the parameter $r$ will be further clarified in Section~\ref{sec:ExpData}.
	
	Now, we focus on an important result which establishes connections between the macroscopic closures %, see \cite{greenshields}, %, see e.g. \cite{GreenshieldsSymposium,Rosini},
	and the generic probability of accelerating $P$ used at the microscopic level in the present kinetic framework.
	\begin{theorem}\label{th:Greenshields}
		Let $r=1$, that is consider only continuous asymptotic kinetic distributions \eqref{eq:SteadySolution}. In the limit $\sigma^2 \to 0$, the kinetic model \eqref{eq:GeneralFP} with desired speeds \eqref{eq:DesiredSpeeds1} provides the macroscopic closure
		\begin{equation*}
			u(\rho) = \vm P(\rho), \quad 0 \leq \rho \leq \rhom,
		\end{equation*}
		where $P:[0,\rhom] \to [0,1]$ is a probability function.
	\end{theorem}
	\begin{proof}
		In order to prove the statement we have to verify that $u=\vm P(\rho)$ is solution of the non-linear equation $R(u)=f^\infty(u^-) R_A(u) - f^\infty(u^+) R_B(u)$ for $\sigma^2 \to 0$, with $R_A(u)$ and $R_B(u)$ defined in \eqref{eq:Case1R}. Substituting $u=\vm P(\rho)$ in the expression of $R_A(u)$ we find
		\[
			R_A(\vm P(\rho))=\frac{\vm^2(1-P(\rho))^2}{(c^A-2)(c^A-1)}-\frac{\vm^2(1-P(\rho))^{c^A}}{(c^A-2)(c^A-1)}-\frac{\vm^2 P(\rho) (1-P(\rho))^{c^A}}{c^A-1}.
		\]
		Since $c^A=\frac{2}{\sigma^2 P(\rho)}+2 \to \infty$ for each fixed $\rho$ when $\sigma^2 \to 0$, we have that
		\[
			\frac{\vm^2(1-P(\rho))^2}{(c^A-2)(c^A-1)} \xrightarrow{\sigma^2 \to 0} 0,
		\]
		while
		\[
			\frac{\vm^2(1-P(\rho))^{c^A}}{(c^A-2)(c^A-1)} \xrightarrow{\sigma^2 \to 0} 0, \quad \frac{\vm^2 P(\rho) (1-P(\rho))^{c^A}}{c^A-1} \xrightarrow{\sigma^2 \to 0} 0
		\]
		because $0 \leq 1-P(\rho) \leq 1$. Therefore, $R_A(\vm P(\rho)) \to 0$ in the limit $\sigma^2 \to 0$. The same considerations can be applied for the braking part. In fact, substituting $u=\vm P(\rho)$ in $R_B(u)$ we find
		\begin{align*}
			R_B(\vm P(\rho)) =& \frac{\vm^2 P(\rho)^2 (1-P(\rho))^2}{(c^B-1)(c^B-2)} - \frac{\vm^2 P(\rho)^{c^B} (1-P(\rho))^2 (1+P(\rho))^{2-c^B}}{(c^B-1)(c^B-2)} \\
			& - \frac{\vm^2 P(\rho)^{c^B} (1-P(\rho))^2 (1+P(\rho))^{1-c^B}}{c^B-1}.
		\end{align*}
		Recall that $c^B=\frac{2}{\sigma^2}+2$ and thus $c^B \to \infty$ when $\sigma^2 \to 0$. While, since $1 \leq 1+P(\rho) \leq 2$ and $1-c^B < 0$, one verifies that $1+P(\rho)^{2-c^B} \to 0$ and $1+P(\rho)^{1-c^B} \to 0$. These facts guarantee that $R_B(\vm P(\rho)) \to 0$, in the limit $\sigma^2 \to 0$ and the thesis is proved.
	\end{proof}
	
	\begin{figure}[t!]
		%\centering
		\includegraphics[width=0.52\textwidth]{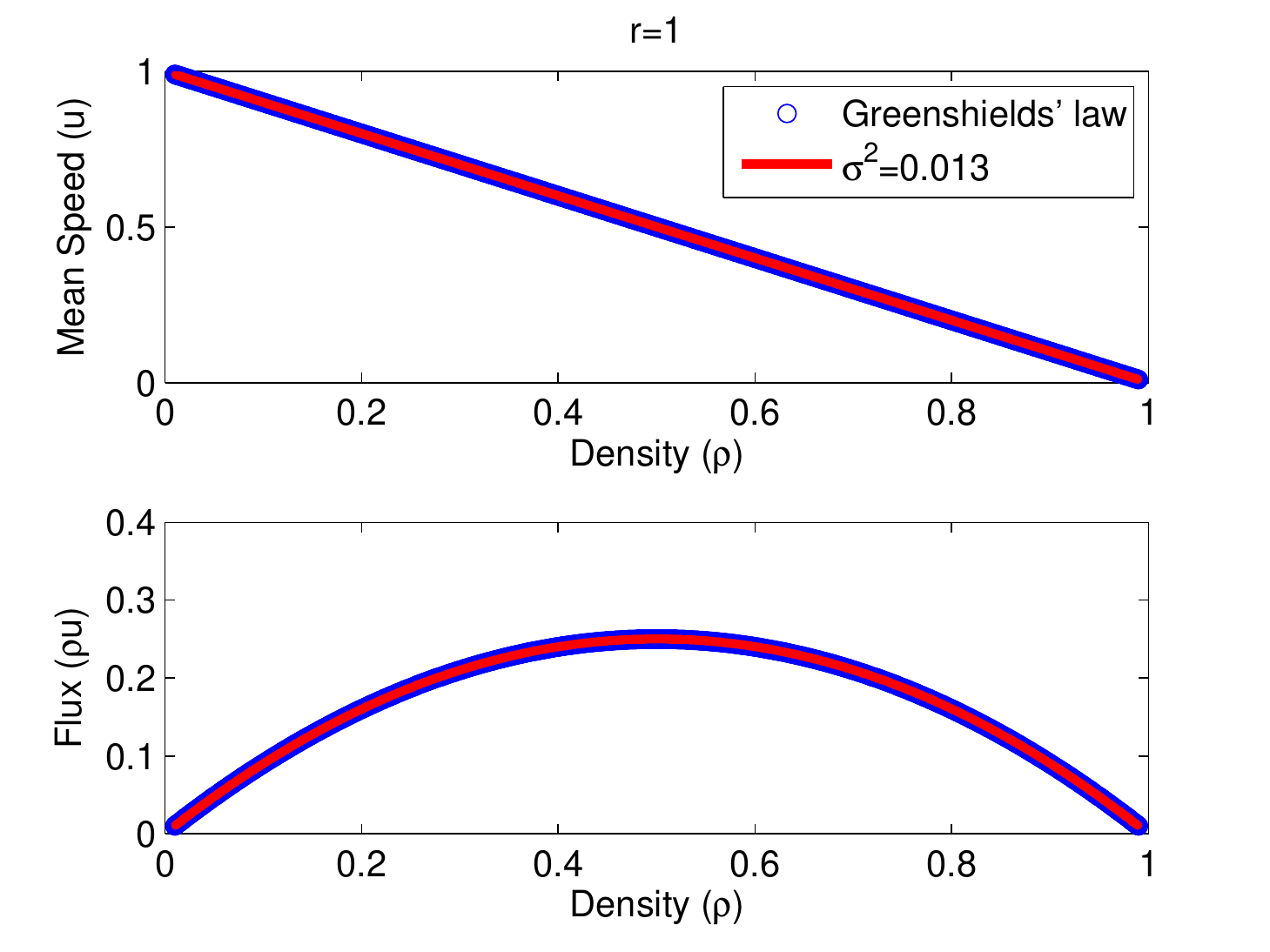}
		%\hfill
		\includegraphics[width=0.52\textwidth]{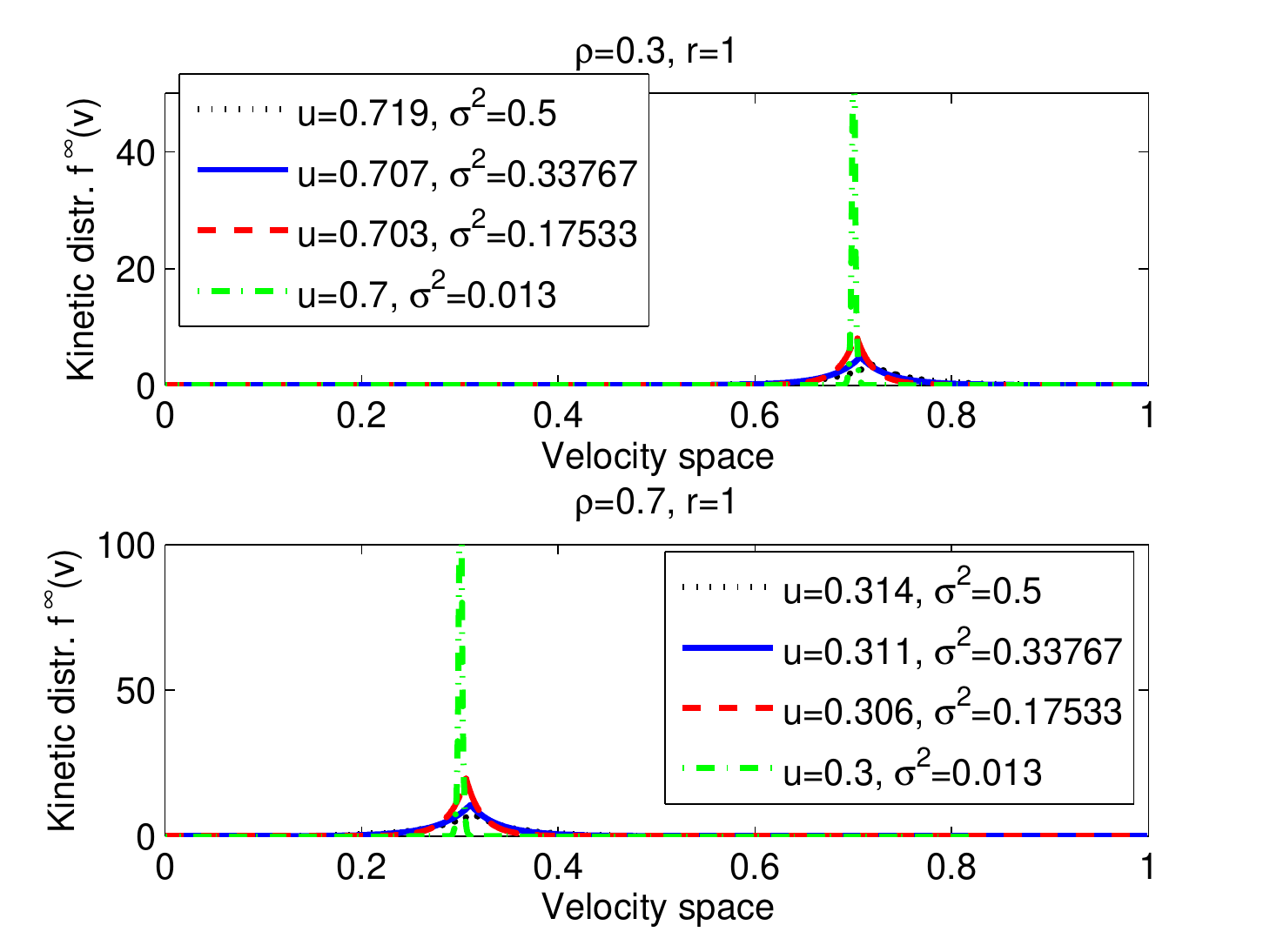}
		\caption{Left panels: diagrams of traffic obtained with the Greenshields' closure law (blue data) and the kinetic model with desired speeds $V_A=v+P(\vm-v)$ and $V_B=Pu$, fixing $r=1$ and $\sigma^2=0.013$. Right panels: the equilibrium distributions for decreasing values of the variance $\sigma^2$ and $\rho=0.3$ (top), $\rho=0.7$ (bottom).\label{fig:Case1vsGreenshields}}
	\end{figure}
	
	The previous Proposition ensures that an empirical closure law $u(\rho)$ between the mean speed and the density (see e.g. \cite[Chap.\,10.1]{Rosini}) can be derived from a kinetic approach based on a microscopic model in which the rules are given by \eqref{eq:MicroRules}, with suitable choices of the probability of accelerating $P$. In other words, in this framework the probability $P$, which influences the microscopic behavior of drivers, can explain also the macroscopic trend of the flow when $\sigma^2 \to 0$.
	
	This relation is investigated numerically in Figure \ref{fig:Case1vsGreenshields} for the Greenshields' closure  frequently used in macroscopic traffic models, see \cite{greenshields}. Let $\rhom=\vm=1$. In the left panels, we show the diagrams of traffic provided by the Greenshields' law $u(\rho)=1-\rho$ (blu data) and by the Fokker-Planck model discussed in this paragraph (red data) with $\sigma^2=0.013$. The convergence rate of the equilibrium speeds to the relation $u(\rho)=1-\rho$ when $\sigma^2 \to 0$ is given in Table \ref{tab:error} in which the distance is computed as $\norm{(1-\boldsymbol{\rho})-\mathbf{u^{\mathsf{eq}}}}_2$, where $\boldsymbol{\rho}$ is a vector of densities and $\mathbf{u^{\mathsf{eq}}}$ is a vector of the equilibrium values found by solving $R(u)=0$ for each density in $\boldsymbol{\rho}$. One can observe that the error goes to zero as $\sigma^2 \to 0$ with sub-linear rate.
	\begin{table}[!t]
		\begin{center}
			\begin{tabular}{l|c|c}
				$\sigma^2$ & Distance (2-norm) & Rate \\
				\hline
				\hline
				$0.5$ & $4.4872E-01$ & - \\                          
				$0.25$ & $1.8192E-01$ & $1.0672$ \\
				$0.125$ & $9.2778E-02$ & $0.7131$ \\
				$0.0625$ & $4.7283E-02$ & $0.7279$\\
				$0.03125$ & $2.3873E-02$ & $0.7491$ \\
				$0.015625$ & $1.1995E-02$ & $0.7602$\\
				\hline
			\end{tabular}
		\end{center}
		\caption{Table of the distance $\norm{(1-\boldsymbol{\rho})-\mathbf{u^{\mathsf{eq}}}}_2$ between the Greenshields' equilibrium speeds $1-\boldsymbol{\rho}$ and th equilibrium ones $\mathbf{u^{\mathsf{eq}}}$ of the kinetic model with desired speeds \eqref{eq:DesiredSpeeds1}. \label{tab:error}}
	\end{table}
	In the right panels of Figure \ref{fig:Case1vsGreenshields} we plot the continuous stationary kinetic distributions \eqref{eq:Case1} for %$\kappa=1$,
	$\rho=0.3$ (top plot), $\rho=0.7$ (bottom plot) and several decreasing values of the variance $\sigma^2$. As $\sigma^2 \to 0$, the equilibrium distributions approach a Dirac delta centered in $u(\rho)=1-\rho$, namely all vehicles travel with same speed at equilibrium. This fact is not surprising because when $\sigma^2 \to 0$ the microscopic rules \eqref{eq:MicroRules} become purely deterministic and the diffusion operator in the Fokker-Planck equation \eqref{eq:GeneralFP} vanishes. % and therefore we are in the case of synchronized flow.%Again, the previous consideration permits to establish a connection from the closure of the present kinetic model and the Greenshields' law, because this latter is used as closure of macroscopic models in which there is no information on the actual velocity of a given car.
	
\subsection{Case 2}\label{sec:Case2} Here we assume
	\begin{equation} \label{eq:DesiredSpeeds2}
		V_A=\min\{v+\Dv,\vm\}, \quad V_B=Pu.
	\end{equation}
	Thus we do not modify the desired speed in braking with respect the previous case. Instead, in acceleration, we assume that the desired speed is not influenced by the congestion level of the flow. More precisely, we model the situation in which, once a driver decides to accelerate, the desired speed is reached with a fixed jump of velocity, so that $V_A$ is $v+\Dv$ if the resulting speed is smaller than $\vm$, while $V_A=\vm$ when $\vm-v<\Dv$. In this case the amplitude of the jump is less than $\Dv$. Observe that this choice was introduced in \cite{PgSmTaVg2}. The parameter $\Dv$ would be a finite parameter which models the physical velocity jump performed in acceleration. Clearly, it may depend on the mechanical microscopic characteristics of vehicles, see \cite{PgSmTaVg3}, but here we will assume that it is constant. With this choice and evaluating explicitly the integrals in \eqref{eq:SteadySolution}, the asymptotic distribution $f^\infty$ writes as
	\begin{subequations}
		\label{eq:Case2}
		\begin{align}
			f^\infty(v) &= f^\infty(u^-)
			\begin{cases}
				\exp\left( \frac{c-2}{\Dv}(v-u) \right), & v<u<\vm-\Dv \\[2ex]
				\left( \frac{\vm-u}{\Dv} \right)^{c} \exp\left( \frac{c-2}{\Dv} (v+\Dv-\vm) \right), & v<\vm-\Dv<u \\[2ex]
				\left( \frac{\vm-u}{\vm-v} \right)^{c}, & \vm-\Dv<v<u
			\end{cases} , \quad \text{if \; $v < u$} \label{eq:Case2Acc}\\
			f^\infty(v) &= f^\infty(u^+) \left( \frac{u-Pu}{v-Pu} \right)^{c}, \quad \text{if \; $v > u$} \label{eq:Case2Brak}
		\end{align}
	\end{subequations}
	where now $c:=c^A=c^B=\frac{2}{\sigma^2}+2$. The kinetic distribution $f^\infty$ given in \eqref{eq:Case2Acc} for $v < u$ can be computed by studying separately the cases $\min\{v+\Dv,\vm\}=v+\Dv$ and $\min\{v+\Dv,\vm\}=\vm$.
	
	Recall that the non-linear equation $R(u)=0$ is determined by the constraint \eqref{eq:OneConstraint}. Since the steady distributions \eqref{eq:Case1} and \eqref{eq:Case2} differ only for $v<u$, $R_B(u)$ is already defined by \eqref{eq:Case1RB} and we need to recompute only $R_A(u)$. For the steady solution \eqref{eq:Case2Acc}, $R_A$ is defined as
	\begin{equation*}
		R_A(u) =  f^\infty(u^-)
		\begin{cases}
				\int_0^u (u-v) \exp\left( \frac{c-2}{\Dv}(v-u) \right) \mathrm{d}v, & \text{if $u<\vm-\Dv$}\\[3ex]
				\int_0^{\vm-\Dv} (u-v) \left( \frac{\vm-u}{\Dv} \right)^{c} \exp\left( \frac{c-2}{\Dv} (v+\Dv-\vm) \right) \mathrm{d}v \\[1ex] + \int_{\vm-\Dv}^u (u-v) \left( \frac{\vm-u}{\vm-v} \right)^{c} \mathrm{d}v, & \text{if $u>\vm-\Dv$}
		\end{cases}.
	\end{equation*}
	Computing explicitly the integrals, we find
	\begin{equation}\label{eq:Case2RA}
		R_A(u) =
		\begin{cases}
			\left( \frac{\Dv}{c-2} \right)^2 \left(1-\exp(-\frac{(c-2)}{\Dv}u)\right) - \frac{\Dv}{c-2} u \exp(-\frac{(c-2)}{\Dv}u), & \text{if $u<\vm-\Dv$}\\[3ex]
			\frac{\Dv}{c-2}\left( \frac{\vm-u}{\Dv} \right)^{c} \left[ u+\Dv-\vm - u \exp(-\frac{c-2}{\Dv}(\vm-\Dv)) \right] \\[1ex]
			+ \left( \frac{\Dv}{c-2} \right)^2 \left( \frac{\vm-u}{\Dv} \right)^{c} \left[ 1-\exp(-\frac{c-2}{\Dv}(\vm-\Dv)) \right] & \text{if $u>\vm-\Dv$} \\[1ex]
			+ \frac{(\vm-u)^{c}}{c-1} \left[ (\vm-u-\Dv) \Dv^{1-c} + \frac{(\vm-u)^{2-c}}{c-2} - \frac{\Dv^{2-c}}{c-2} \right].
		\end{cases}
	\end{equation}

	\begin{figure}[t!]
		%\centering
		%\includegraphics[width=0.52\textwidth]{Case2FD-sigma}
		\includegraphics[width=0.52\textwidth]{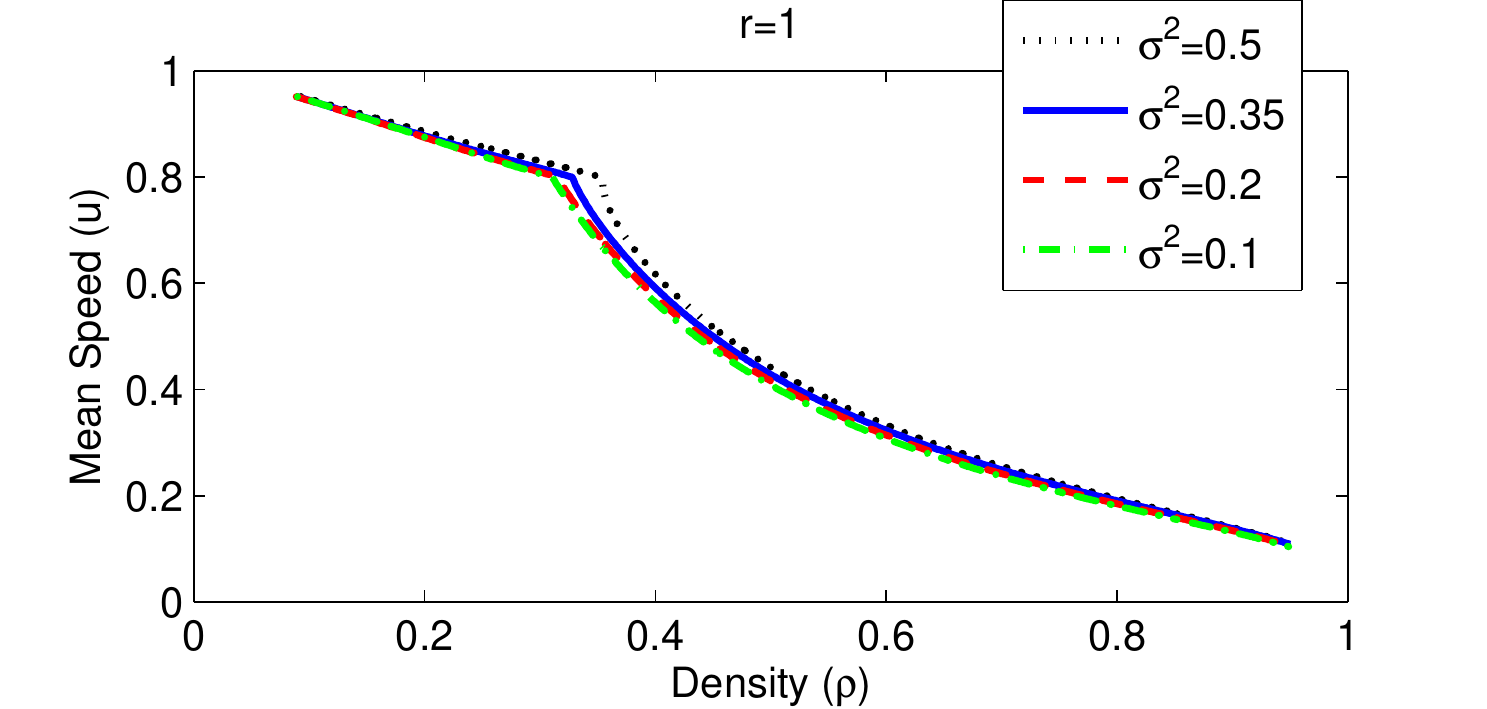}
		\includegraphics[width=0.52\textwidth]{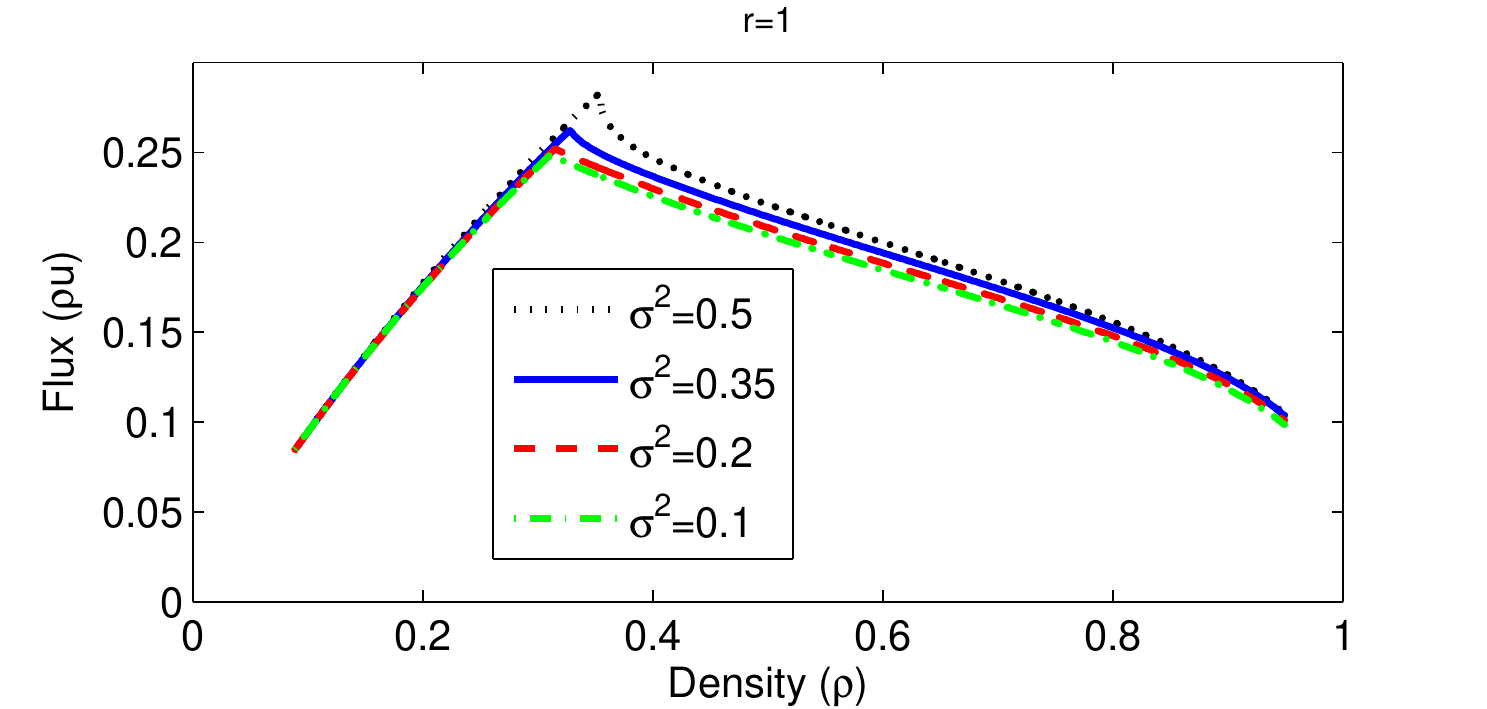}
		\caption{Macroscopic diagrams of traffic using the desired speeds \eqref{eq:DesiredSpeeds2}. We take $r=1$, $\Dv=0.2$, for several values of the variance $\sigma^2$.\label{fig:Case2FD-sigma}}
	\end{figure}
	
	\begin{figure}[t!]
		%\centering
		\includegraphics[width=0.52\textwidth]{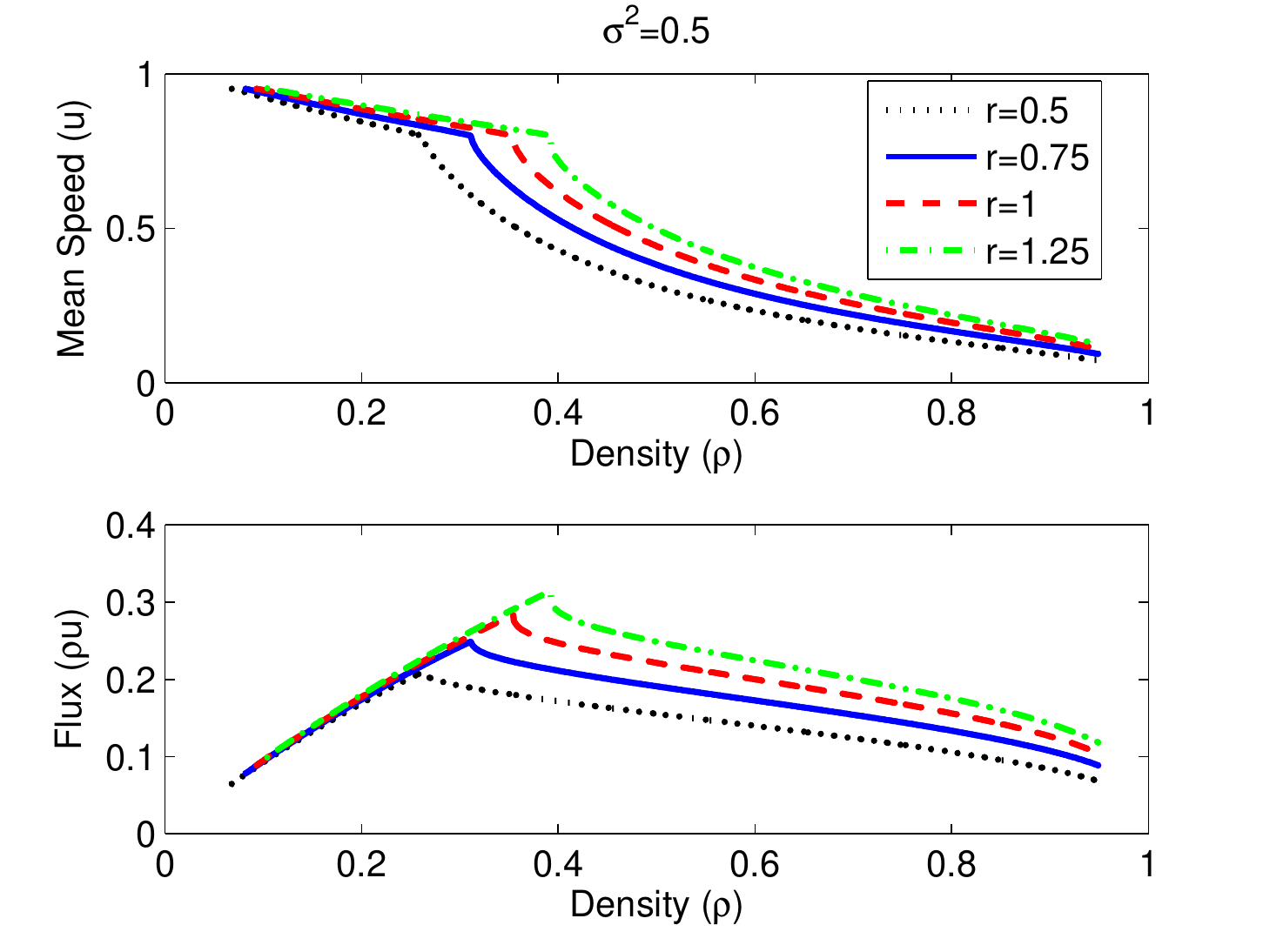}
		%\hfill
		\includegraphics[width=0.52\textwidth]{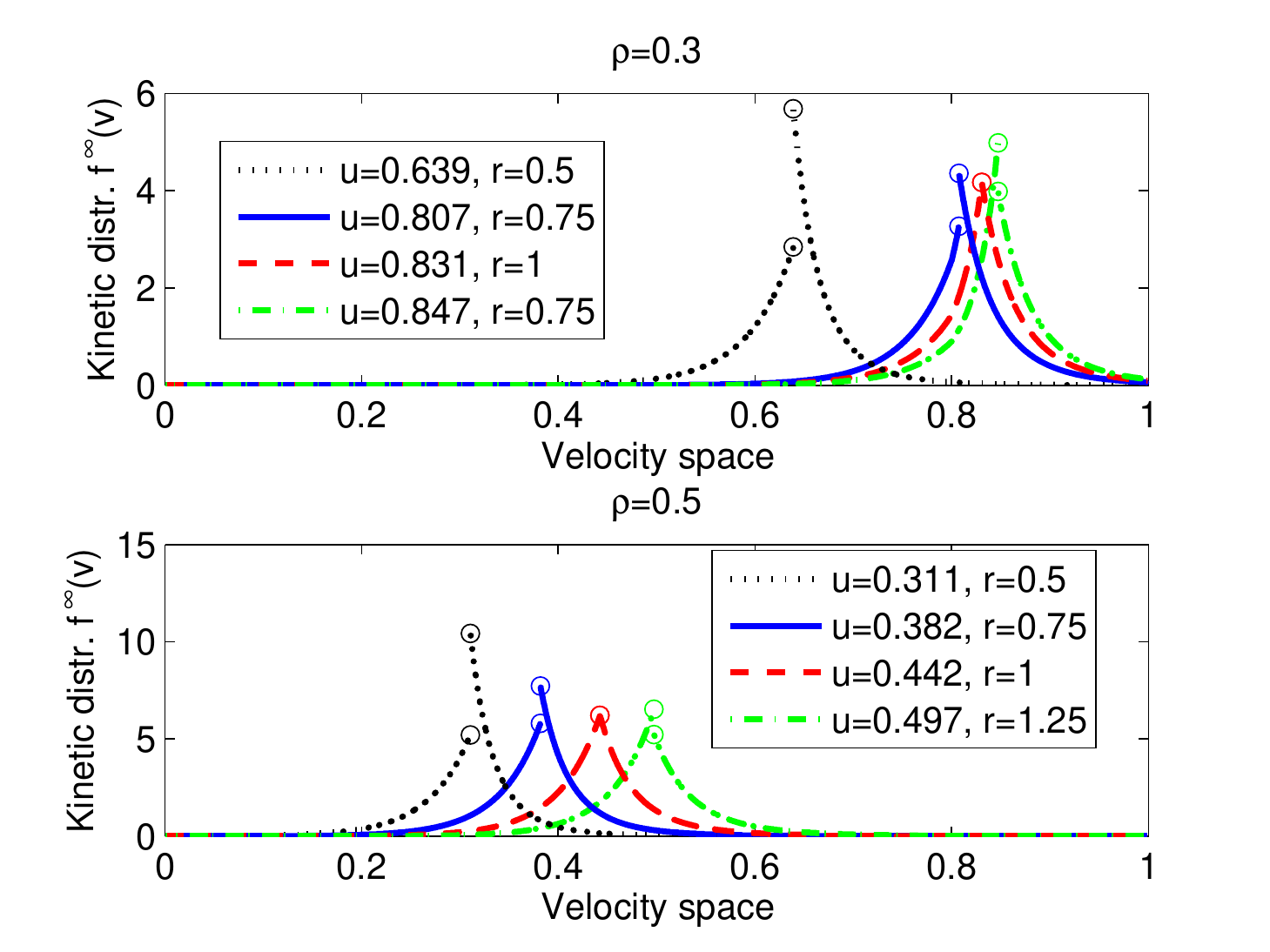}
		\caption{Left panels: macroscopic diagrams of traffic using the desired speeds \eqref{eq:DesiredSpeeds2}. We take $\sigma^2=0.5$, $\Dv=0.2$ and we study the dependence on the parameter $r$. Right panels: equilibrium distributions for $\rho=0.3$ (top) and $\rho=0.5$ (bottom). The circles define the values $f(v)$ for $v\to u^-$ and $v\to u^+$.\label{fig:Case2FD-r}}
	\end{figure}
	
	In Figure~\ref{fig:Case2FD-sigma} we study the influence of %$\kappa$ (left panels) and $\sigma^2$ (right panels)
	$\sigma^2$ on the mean speed and on the flux for $r=1$, that is allowing only continuous steady states. As already discussed, the equilibrium values are obtained by solving numerically $R(u)=0$. The top plots show the speed-density diagrams, while at the bottom we have the flux-density diagrams. In all simulations we set $\rhom = \vm = 1$. The probability of accelerating $P$ is taken as prescribed in \eqref{eq:Plaw}. The parameter $\Dv$ is chosen as $0.2$. %The parameter $\kappa$ seems to influence strongly the absolute value of the slope of the flux curves at high densities. In fact, it becomes smaller, and thus we have a smoother evolution towards the jammed traffic flow, as $\kappa$ decreases. Instead, as
	As in Case 1, the variance $\sigma^2$ seems to produce only small variations in the equilibrium values of the macroscopic speed and of the flux. Only the values near to the phase transition are affected, in particular they become sharper when $\sigma^2$ increases. Again, a multivalued behavior cannot be obtained by acting on %$\kappa$ and
	$\sigma^2$.
	
	Next we consider the diagrams obtained with %$\kappa=1.5$ and
	$\sigma^2=0.5$ but depending on the parameter $r$ which allows one to take into account different steady states according to their regularity at $v=u$, see equation \eqref{eq:r}. The right panels of Figure \ref{fig:Case2FD-r} show the stationary distributions for $\rho=0.3,0.5$. In the left panels of Figure \ref{fig:Case2FD-r}, instead, we show that now we obtain a more significant and accentuated dispersion of the equilibrium values in a larger range of densities. The dispersion of data becomes small at low and high densities and this is in accordance with experimental data.
	
	We stress the fact that the kinetic theory allows us to explain both the theoretical and the experimental fundamental diagrams in terms of microscopic rules. In fact, on the one hand, with the the desired speeds \eqref{eq:DesiredSpeeds1} we reproduce the heuristic closure laws used in the macroscopic modeling, as for instance the Greenshields' law, see Proposition \ref{th:Greenshields}. On the other hand, the desired speeds \eqref{eq:DesiredSpeeds2} allow us to reproduce well experimental data.
	
	Let us focus first on the free-flow phase. This regime occurs at low densities, i.e when the road is not congested, thus there is a large distance among the vehicles, the interactions are rare and the velocity of the flow can be high. In Figure \ref{fig:Case2FD-r} we observe that the flux increases nearly linearly with respect to the density of vehicles and the dispersion of the data is small. Nevertheless, the average velocity of the vehicles does not coincide with the maximum allowed speed. In fact, looking at the speed-density diagrams, we note that the macroscopic speed decreases as the density increases. This means that the model does not provide stationary solutions being Dirac delta centered in the maximum speed $\vm$ when the density is less than a fixed value, see e.g. the quantized steady state distributions analytically computed in \cite{PgSmTaVg2}. This result seems to be coherent with the experience because one tends to travel at the maximum speed only in really free road conditions while the velocity is reduced as the number of vehicles on the road increases. However, observe that, as we expect, the decrease of the mean speed is not so fast to cause a decrease of the average flux.
	
	Now, let us focus on the congested phase of traffic. In this regime the road becomes progressively jammed, thus the interactions among the vehicles are more frequent and drivers tend to control their velocity with respect the speed of the flow  which is forced to decrease due to the large number of vehicles. As a consequence, in Figure \ref{fig:Case2FD-r} we observe that the flux decreases as the density of the vehicles increases because of a drastic reduction of the mean speed. Moreover, there is a range of densities in which it is evident that we can obtain scattered values in the flux-density plane. However, as the road becomes jammed, the possible scattering reduces because the distance between the equilibrium curves reduces, see for instance the speed-density diagram in Figure \ref{fig:Case2FD-r}. This is consistent with the daily experience of driving on highways, because in congested flow all vehicles tend to travel at the same speed.
	
	Finally, notice that, for increasing values of $r$, the diagrams provided by the present kinetic model reproduce naturally the phenomenon of the capacity drop, that is the sharp decrease of the flux across a density value. In fact, it is possible to identify the so-called critical density in which there is a phase transition between the free and the congested regimes.
	
	\subsection{The role of $\boldsymbol{r}$ and comparisons with data} \label{sec:ExpData}

	Since the aim of this paper is to investigate the dependence between the structure of fundamental diagrams and the microscopic dynamics, in this section we propose comparisons with experimental data. To this end, we also analyze the role of the free parameter $r$, which parametrizes the family of equilibrium distributions $f^\infty$, and we show that it is the key to reproduce multivalued diagrams.
	
	We recall from equation \eqref{eq:r} that $r$ is not easily recognized as a physical parameter, because it gives information on the regularity of $f^\infty$ at $v=u$. In fact, $r$ is the ratio between the two integration constants arising from the solution of the ODE \eqref{eq:ODE}. Therefore, we need to give a recipe in order to choose the parameter and to understand whether we can associate different values of $r$ to the same flow.
	
	To this end, in the following, although it seems that the scattering of data is due only to a mathematical fact, we prove that $r$ can actually be linked to observable/measurable macroscopic quantities of the flow. Start from the quantity $R(u)$ given in \eqref{eq:R}. We observe that it is a function also of the density $\rho$ and of the parameter $r$ when they are not fixed a priori. In fact, as discussed at the beginning of Section \ref{sec:FundamentalDiagrams}, applying the constraint \eqref{eq:ZeroConstraint} and assuming that $f^\infty$ is parametrized by the ratio \eqref{eq:r}, we find the expressions \eqref{eq:IntConst} of the two integration constants. Then, the quantity $R$ given in \eqref{eq:R} can be written as
	\[
		R(u,\rho,r)= \frac{r\rho}{r\rho^-+\rho^+} R_A(u,\rho) - \frac{\rho}{r\rho^-+\rho^+} R_B(u,\rho) 
	\]
	where $R_A$ and $R_B$ are defined by \eqref{eq:Case1RA}-\eqref{eq:Case1RB} and \eqref{eq:Case2RA}-\eqref{eq:Case1RB} when the desired speeds are chosen as \eqref{eq:DesiredSpeeds1} and \eqref{eq:DesiredSpeeds2}, respectively. By Implicit Function Theorem we have that $R(u,\rho,r)=0$ defines implicitly a function
	\begin{equation} \label{eq:rVSmacro}
	r=r(u,\rho)=\frac{R_B(u,\rho)}{R_A(u,\rho)}
	\end{equation}
	in each neighborhood of $(u_0,\rho_0,r_0)$ such that $R(u_0,\rho_0,r_0)=0$ provided that  $u_0 \neq 0$, $u_0 \neq \vm$, $\rho_0 \neq 0$ because in this case one has $\partial_r R(u,\rho,r) \neq 0$. The function \eqref{eq:rVSmacro} thus establishes a link between the free parameter $r$ and the macroscopic quantities of traffic.
	
	Experimental data show several values of the mean speed $u$ at equilibrium for the same value of the density $\rho$, correspondingly the relation \eqref{eq:rVSmacro} states that, for any fixed $\rho$, there exist more values of $r$ related to the same flow. Therefore, the analytical origin of the multivalued region is explained here by the presence of a family of equilibrium solutions parametrized by $r$. The parameter $r$ is the degree of freedom (in addition to $\rho$) being necessary to equilibrium in order to describe the dispersion of data. In other words, we can think that $r$ synthesizes the different properties of the flow (as the composition of the traffic flow or the structural characteristics of the road) which induce different macroscopic dynamics for a given value of the density $\rho$. For instance, comparing this result with our multi-population models \cite{PgSmTaVg3,PgSmTaVg}, we can suppose that $r$ takes into account the heterogeneous composition of the vehicles on the road.
	
	In Figure \ref{fig:Comparisons} we compare the diagrams obtained using the desired speeds \eqref{eq:DesiredSpeeds2} and the experimental measurements provided on a USA highway by the Minnesota Department of Transportation in 2003 (left panel), kindly granted by Seibold et al. \cite{seibold2013NHM}, and in Viale del Muro Torto in Rome (Italy) (right panel), see the review \cite{piccoli2009ENCYCLOPEDIA}. We focus only on the modeling \eqref{eq:DesiredSpeeds2} of the desired speeds because the structure of the diagrams resulting from this choice is totally in agreement with experimental data. In order to choose $r$ using the data we first individuate a density $\rho$ and we select from the experimental diagram all values of the flux $q$ at equilibrium corresponding to the fixed $\rho$. Then we compute the mean speed $u=q/\rho$, so that we can evaluate the relation \eqref{eq:rVSmacro} for each pair $(u,\rho)$ in order to determine $r$. Finally, we plot the corresponding \textquotedblleft theoretical\textquotedblright \ diagrams.
	
	\begin{figure}[t!]
		%\centering
		\includegraphics[width=0.44\textwidth]{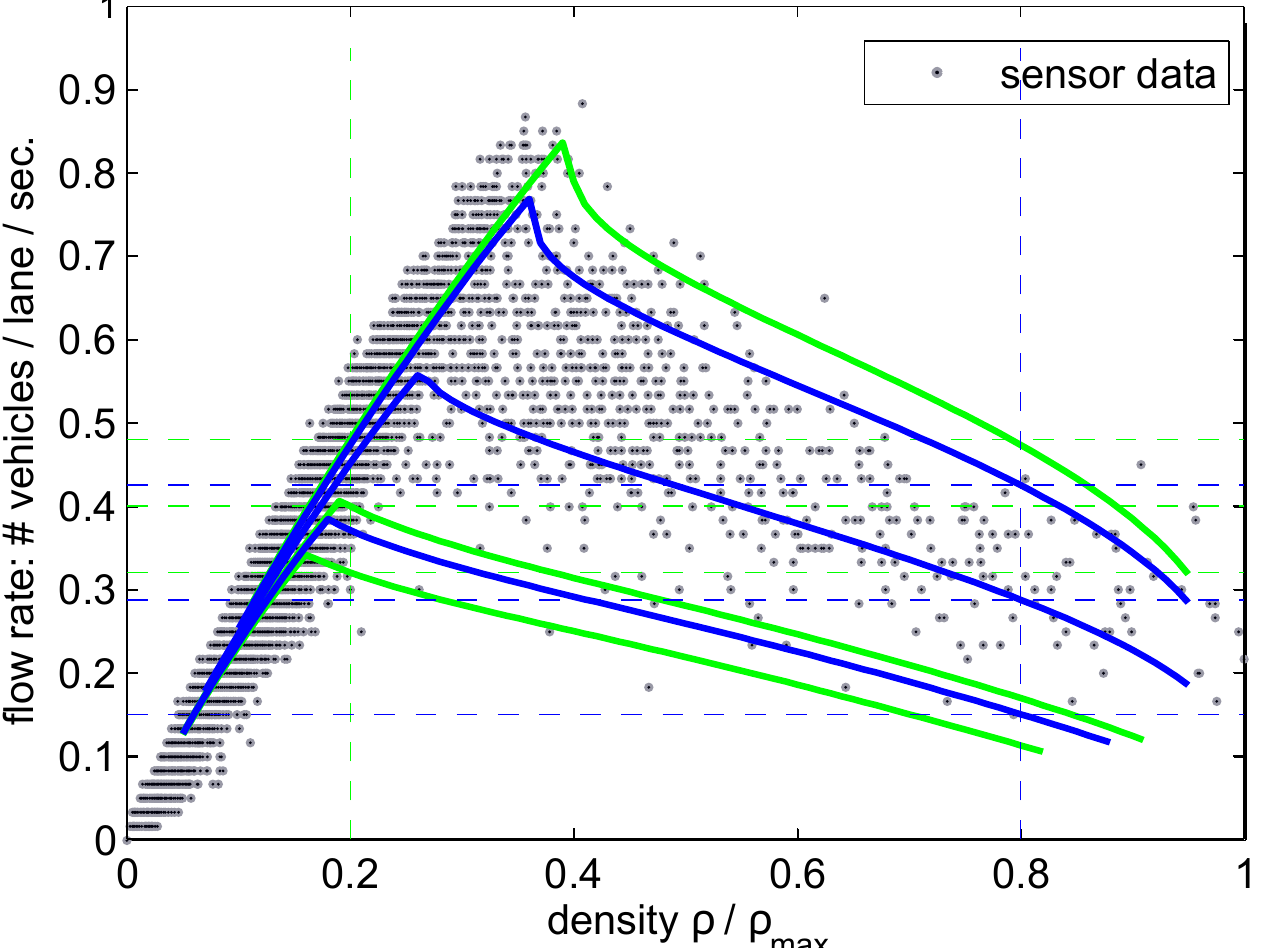}
		\hfill
		\includegraphics[width=0.56\textwidth]{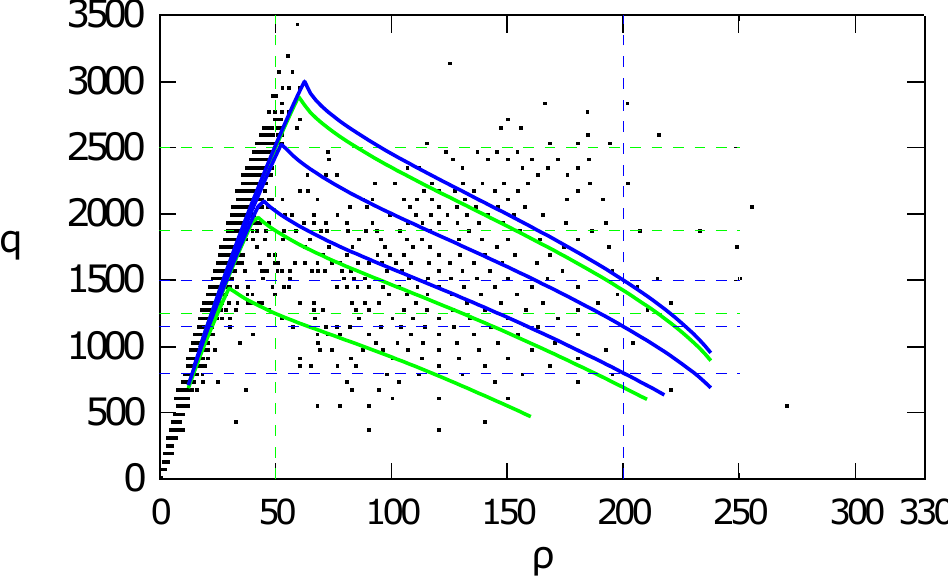}
		\caption{Comparison between the diagrams provided by the present kinetic model with desired speeds \eqref{eq:DesiredSpeeds2} and the experimental diagrams provided by kind permission of Seibold et al. \cite{seibold2013NHM} (left panel) and experimental data collected in one week in Viale del Muro Torto, Rome (Italy),
		from \cite{piccoli2009ENCYCLOPEDIA} (right panel). The parameters of the model are set as $\sigma^2=0.5$ and $\Dv=0.2$. The green curves are related to the density $\rho=0.2$, while the blue curves are related to the density $\rho=0.8$. The free parameter $r$ is computed using the relation \eqref{eq:rVSmacro}.\label{fig:Comparisons}}
	\end{figure}
	
	More precisely, in both panels of Figure \ref{fig:Comparisons} we fix two values of the normalized density $\rho=0.2, 0.8$, taking as maximum density in the right panel $\rhom=250$ vehicles per kilometer. For each of the two densities we consider three flux values, the maximum $q_1=\max(q(\rho))$, the minimum $q_2=\min(q(\rho))$ and finally the mean value $q_3=(q_1+q_2)/2$. The intersections of the green dashed lines individuate the pairs of data $(\rho=0.2,q_i)$, $i=1,2,3$, while the intersections of the blue dashed lines individuate the pairs of data $(\rho=0.8,q_i)$, $i=1,2,3$. For each pair of data the normalized mean speed is computed as $u = q/(\rho u_{\max})$ where $u_{\max}$ is the maximum speed and is approximated with the linear interpolation of the free flow regime. Finally, the pairs $(u,\rho)$ define three values of the parameter $r$ for each fixed density by means of equation \eqref{eq:rVSmacro}. The green diagrams are related to $\rho=0.2$, while the blue diagrams are related to $\rho=0.8$.
	We observe that there is a remarkable correspondence between the \textquotedblleft theoretical\textquotedblright \ and the experimental diagrams. In fact, the equilibrium curves provided by the present kinetic model reproduce on the whole the structure of the experimental data, included the scattering, both in the free and the congested regime. This means that, in addition to explaining the dispersion with the variability of $r$, the microscopic model is already endowed with the other physical characteristics of traffic, included the phase transition. We stress that we do not interpolate the data. We have only one degree of freedom and the remaining features are explained by the model itself.
	
	\begin{figure}[t!]
		\centering
		\includegraphics[width=0.49\textwidth]{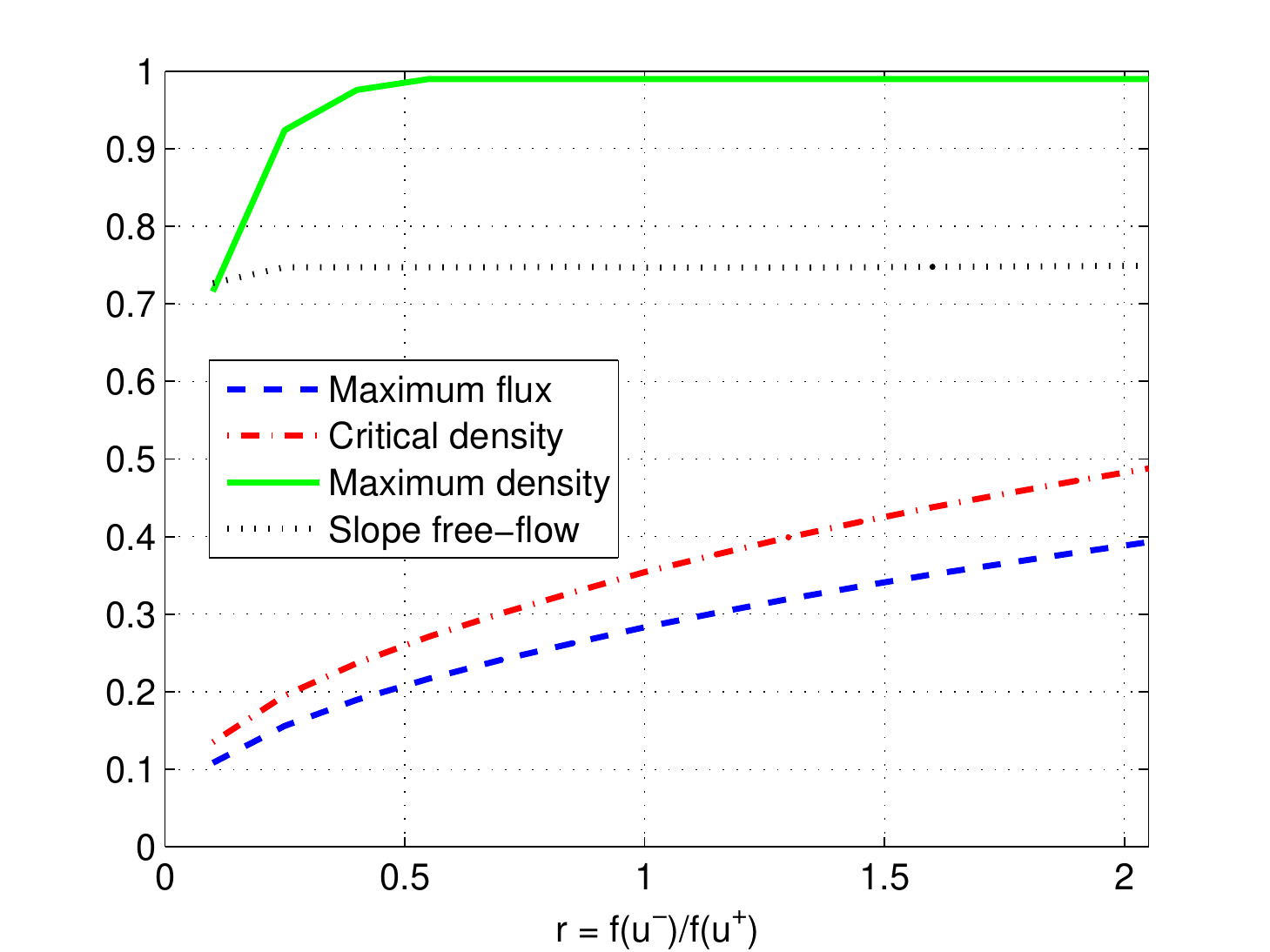}
		\caption{Monotone relation between the ratio $r=\frac{f^\infty(u^-)}{f^\infty(u^+)}$ and the maximum flux (blue dashed line), the critical density (red dot-dashed line), the maximum density (green solid line), the slope in the free-flow regime (black dotted line).\label{fig:rVSmacro}}
	\end{figure}
	
	Finally, in Figure \ref{fig:rVSmacro} we show the relation holding between the parameter $r$ and some macroscopic quantities: precisely, the maximum value of the flux (blue dashed line), the value of the critical density (red dash-dotted line), the maximum value of the density (green solid line) and the slope of the free flow branch (black dotted line). We observe an increasing monotone dependence which allows us to state that only a limited range of values for $r$ can be admissible. In fact, for example, for larger $r$ the critical density becomes too high when compared to that provided by experimental data.

\section{Conclusions and Outlook} \label{sec:Conclusions}
%\todo[inline]{Future perspective: derive a multipopulation model using the interaction rules given in the present framework. In this way, since vehicles interact with respect the mean speed $u$, the equations of the model are coupled only on $s=\sum_i \rho_i l_i$ and thus we still have only one collision operator. Reduced complexity of the model; multivalued diagrams without assuming discontinuous steady states; richer time-asymptotic solutions which are not quantized.}

In this paper, starting from interaction rules %based on two levels of stochasticity and
derived from those introduced in \cite{HertyPareschi10} and in \cite{PgSmTaVg2}, we have computed a Fokker-Planck model for traffic flow as grazing collision limit of a Boltzmann-type model. In particular, we have aimed at recovering multivalued fundamental diagrams and at analyzing the influence of the microscopic interactions on the collective dynamics of the flow. In fact, using the Fokker-Planck approximation, we have computed the asymptotic time-distribution of model for the general set of the rules which are defined once the desired speeds are chosen. In particular, we have considered desired speeds preserving the synchronized traffic states, or reproducing for example the Greenshields' closure law \cite{greenshields}, or leading to fundamental diagrams which have the same qualitative properties of the experimental data.

The mathematical key allowing for the scattering of the flux values is the existence of a one-parameter family of stationary distributions. We have shown that the positive real parameter which defines the steady states has a strong link with the macroscopic properties of the vehicular flow. This means that it synthesizes all the properties of the flow which determine a change in the flux for the same density of vehicles. Moreover, conversely to \cite{PgSmTaVg2} in which we aimed at establishing connections between continuous- and discrete-velocity kinetic models using microscopic binary interactions, here the model does not provide quantized asymptotic distributions and thus we obtain a continuous description of the microscopic velocity distribution.

Since the Fokker-Planck approximation is less demanding when compared to the full Boltzmann-type model, we think that the natural sequel of this work could be the study of a multi-population model based on the interaction rules prescribed in the present paper. In fact, in the approximation that vehicles are affected only by the average properties of the flow, the binary interaction terms of the Boltzmann equations are replaced  by simpler mean-field terms for all populations of vehicles. Thus, each equation contains only one collision operator and furthermore equations are coupled only through the macroscopic density $\rho$ and the mean speed $u$, a fact which permits to keep the complexity of the model low. On the other hand, we expect to obtain multivalued diagrams as a result of the heterogeneous composition of the flow without allowing for discontinuous stationary distributions because we explicate one of the properties included in the mathematical free parameter of this work.

\paragraph{Acknowledgments} This work was partly supported by \textquotedblleft National Group for Scientific Computation (GNCS-INDAM)\textquotedblright.\\
This work has been supported by KI-Net NSF RNMS grant No. 1107444, grants DFG Cluster of Excellence \textquotedblleft Production technologies for high-wage countries\textquotedblright, HE5386/13,14-1 and 15/1 
as well as the VIGONI--MUIR project.\\
Andrea Tosin acknowledges that this work has been written within the activities of GNFM (Gruppo Nazionale per la Fisica Matematica) of INdAM (Istituto Nazionale di Alta Matematica), Italy.\\
Giuseppe Visconti wishes to thank the RWTH Aachen University (IGPM Institute - LuF Mathematik) for their hospitality and financial support.

\bibliographystyle{plain}
\bibliography{references}

\end{document}